\documentclass[12pt, draftclsnofoot, onecolumn]{IEEEtran}

\usepackage{cite}
\usepackage[pdftex]{graphicx}
\graphicspath{{Figure/}}
\usepackage[cmex10]{amsmath}
\usepackage{amsthm}
\usepackage{amssymb}
\usepackage{multirow}
\usepackage{booktabs}

\usepackage{algorithm}
\usepackage{algpseudocode}

\usepackage{array}
\usepackage{color}
\usepackage{epstopdf}
\usepackage{amsfonts}
\usepackage[acronym,shortcuts]{glossaries}
\usepackage[normalem]{ulem}

\usepackage{pgfplots}
\pgfplotsset{compat=1.3}
\usepackage{tikz}							
\usetikzlibrary{shapes}
\usetikzlibrary{spy}
\usetikzlibrary{circuits}
\usetikzlibrary{arrows}
\usepackage{svg}

\newcommand{\linewidthmod}{0.49\linewidth}

\makeglossaries
\newacronym{ict}{ICT}{information and communications technology}
\newacronym[plural=BSs,firstplural=base stations (BSs)]{bs}{BS}{base station}
\newacronym[plural=PAs,firstplural=power amplifiers (PAs)]{pa}{PA}{power amplifier}
\newacronym{sdr}{SDR}{signal-to-distortion ratio}
\newacronym{snr}{SNR}{signal-to-noise ratio}
\newacronym{sndr}{SNDR}{signal-to-noise-and-distortion ratio}
\newacronym{snidr}{SNIDR}{signal-to-noise-and-interference-and-distortion ratio}
\newacronym{los}{LOS}{line-of-sight}
\newacronym{z3ro}{Z3RO}{zero third-order distortion}
\newacronym{mimo}{MIMO}{multiple input multiple output}
\newacronym{siso}{SISO}{single input single output}
\newacronym{mrt}{MRT}{maximum ratio transmission}
\newacronym{dpd}{DPD}{digital pre-distortion}
\newacronym{ula}{ULA}{uniform linear array}
\newacronym{se}{SE}{spectral efficiency}
\newacronym{ee}{EE}{energy efficiency}
\newacronym{tdma}{TDMA}{time division multiple access}
\newacronym{ran}{RAN}{radio access network}
\newacronym{lte}{LTE}{long term evolution}

\newacronym{psd}{PSD}{power spectral density}
\newacronym{aclr}{ACLR}{adjacent channel leakage ratio}
\newacronym{cpu}{CPU}{central processing unit}
\newacronym{iid}{i.i.d.}{independently and identically distributed}
\newacronym{ue}{UE}{user equipment}
\newacronym{nlos}{NLoS}{non-line-of-sight}
\newacronym{hpbm}{HPBM}{half power beam width}
\newacronym{rts}{RTS}{ray tracing simulator}
\newacronym{ag}{AG}{array gain}
\newacronym{arp}{ARP}{Antenna Reference Point}
\newacronym{ecdf}{eCDF}{Empirical Cumulative Distribution Function}
\newacronym{evm}{EVM}{Error Vector Magnitude}
\newacronym{fft}{FFT}{fast Fourier transform}
\newacronym{ib}{IB}{in-band}
\newacronym{im}{IM}{intermodulation}
\newacronym{mf}{MF}{matched filtering}
\newacronym{mr}{MR}{maximum ratio}
\newacronym{ofdm}{OFDM}{orthogonal frequency division multiplexing}
\newacronym{oob}{OOB}{out-of-band}
\newacronym{papr}{PAPR}{peak-to-average power ratio}
\newacronym{rx}{RX}{Receiver}
\newacronym{trp}{TRP}{Transmission Reception Point}
\newacronym{tx}{TX}{transmitter}
\newacronym{zf}{ZF}{zero forcing}
\newacronym{dl}{DL}{downlink}
\newacronym{rzf}{RZF}{regularized zero forcing}
\newacronym{slc}{SLC}{spatial leakage suppression}
\newacronym{kpi}{KPI}{key performance indicator}
\newacronym{awgn}{AWGN}{additive white gaussian noise}
\newacronym{qam}{QAM}{quadrature amplitude modulation}
\newacronym{amam}{AM/AM}{amplitude modulation to amplitude modulation}
\newacronym{ampm}{AM/PM}{amplitude modulation to phase modulation}
\newacronym{zmcscg}{ZMCSCG}{zero mean circularly symmetric complex Gaussian}

\usepackage{mathtools}
\DeclarePairedDelimiter\ceil{\lceil}{\rceil}
\DeclarePairedDelimiter\floor{\lfloor}{\rfloor}
\DeclarePairedDelimiter\ceilfloor{\lfloor}{\rceil}
\newcommand\round[1]{\left[#1\right]}

\theoremstyle{definition}
\newtheorem{theorem}{Theorem}
\newtheorem{remark}{Remark}
\newtheorem{proposition}{Proposition}
\newtheorem{lemma}{Lemma}
\newtheorem{corollary}{Corollary}

\ifCLASSOPTIONcompsoc
  \usepackage[caption=false,font=normalsize,labelfont=sf,textfont=sf]{subfig}
\else
  \usepackage[caption=false,font=footnotesize]{subfig}
\fi
\usepackage{dblfloatfix}

\begin{document}

\title{Information-Theoretic Study of Time-Domain Energy-Saving Techniques in Radio Access}

\author{Fran\c{c}ois Rottenberg\vspace{-4em}
	\thanks{Fran\c{c}ois Rottenberg is with ESAT-DRAMCO, Ghent Technology Campus, KU Leuven, 9000 Ghent, Belgium (e-mail: francois.rottenberg@kuleuven.be).}
}

\markboth{DRAFT}%
{}
%



\maketitle

\begin{abstract}
	\vspace{-1em}
    Reduction of wireless network energy consumption is becoming increasingly important to reduce environmental footprint and operational costs. A key concept to achieve it is the use of lean transmission techniques that dynamically (de)activate hardware resources as a function of the load. In this paper, we propose a pioneering information-theoretic study of time-domain energy-saving techniques, relying on a practical hardware power consumption model of sleep and active modes. By minimizing the power consumption under a quality of service constraint (rate, latency), we propose simple yet powerful techniques to allocate power and choose which resources to activate or to put in sleep mode. Power consumption scaling regimes are identified. We show that a ``rush-to-sleep" approach (maximal power in fewest symbols followed by sleep) is only optimal in a high noise regime. It is shown how consumption can be made linear with the load and achieve massive energy reduction (factor of 10) at low-to-medium load. The trade-off between \gls{ee} and \gls{se} is also characterized, followed by a multi-user study based on \gls{tdma}.
\end{abstract}

\begin{IEEEkeywords}
	\vspace{-1em}
    Energy consumption, radio access technologies, physical layer, channel capacity. \vspace{-1em}
\end{IEEEkeywords}

%
\IEEEpeerreviewmaketitle


\section{Introduction}

\subsection{Motivation}

As of 2022, yearly data volume has gone up to more than 3 Zettabytes ($10^{21}$ bytes) and the traffic continues to rise at a rate of about 25\%/year~\cite{malmodin_power_2020}. Energy efficiency has improved over the years but not fast enough, which results in an annual energy consumption growth of 2.5\% for the ICT sector~\cite{malmodin_power_2020}. It is becoming increasingly important to reduce energy consumption of wireless communication networks to reach climate ambitions and reduce operational expenses, in other words, ``break the energy curve"~\cite{ekudden_breaking_2020,Andersson2022}. Most energy of wireless networks is consumed by the \gls{ran} and more specifically at base stations~\cite{gruber_earth_2009,3gpp.38.864}. Moreover, the traffic load at a base station is highly varying across the day and most often lightly loaded, with traffic at night being about 10 times lower than during the day~\cite{huawei_green_5G_white_paper}. This opens a big potential for energy reduction through the use of lean transmission techniques that dynamically activate or deactivate resources as a function of the load, letting the system dynamically switch from fully active to deep sleep mode.

\subsection{State of the Art}

Energy-saving techniques are a popular topic. A large effort has been made to integrate these techniques into industrial products and standards. The 5G standard was for instance designed with a lean paradigm in mind which resulted in, \textit{e.g.}, less reference signaling to increase sleep duration~\cite{dahlman20205g,Lopez2022,zhang_fundamental_2017}. At low-to-medium load, a popular scheduling technique is a ``rush-to-sleep" approach which compacts transmission in as few symbols as possible. These symbols are transmitted at maximal power, leaving the remaining symbols in the frame free, so that the sleep duration is maximized. Many studies have been performed to evaluate the gains of such techniques based on standardized power models~\cite{3gpp.38.864}, system-level evaluations~\cite{tombaz_energy_2015,lahdekorpi_energy_2017,matalatala_simulations_2018,frenger_energy_2019,frenger_more_2019} and aided by actual measurements~\cite{golard_evaluation_2022}. The use of machine learning was also identified as an interesting tool to predict traffic and/or to optimize energy-saving features~{\cite{salem_traffic-aware_2019,piovesan_forecasting_2021}. We refer to~\cite{Lopez2022} for a review of these techniques.


Despite much work in the domain, there remains a fundamental gap to be filled by establishing an information-theoretic study of time-domain energy-saving techniques, even for basic systems such as \gls{siso} transceivers. Going back to the underlying physics of energy consumption of base stations and with a proper mathematical formulation of the communication link, a lot of additional understandings and improvements can be obtained: optimization of algorithms, finding optimal power scaling regimes as a function of load, guarantees of optimality for energy-saving features and/or finding the gap from it with existing techniques... As an example, it is not clear if or when a rush-to-sleep approach is optimal or not. We should mention that many recent works have performed this kind of energy-saving studies but they have focused on the spatial domain and more specifically the optimal operation of massive MIMO systems (number of active antennas, served users, power allocation) as a function of the load~\cite{cheng_massive_2015,senel_joint_2019,bjornson_optimal_2015,marinello_antenna_2020,peschiera_linear_2022}. On the other hand, information-theoretic time-domain studies are lacking. The fundamental studies on energy-efficient communications have mainly focused on a stationary transmission of symbols in time at a constant rate and average transmit power $P_{\mathrm{T}}$~\cite{budzisz_dynamic_2014,chen_fundamental_2011,jingon_joung_spectral_2014,wu_green_2016,wu_overview_2017}. Considering an ideal consumption model and a quasi-static channel, this choice seems intuitive. To clarify it, let us formalize the problem. The channel capacity of a complex discrete memoryless \gls{awgn} channel, under an average transmit power constraint $P_{\mathrm{T}}$, is
\begin{align*}
	R&=\log_2\left(1+\frac{P_{\mathrm{T}}}{L\sigma_n^2}\right)=\log_2\left(1+\frac{P_{\mathrm{T}}}{{\sigma}^2}\right)\ \text{[bits/channel use]},
\end{align*}
where $\sigma^2=\sigma_n^2 L$ is the noise power at the receiver $\sigma_n^2$ normalized by the path loss $L$. If the transmission is divided in frames of $N$ symbols, the average rate and transmit power are
\begin{align*}
	R&=\frac{1}{N}\sum_{n=0}^{N-1}\log_2\left(1+\frac{p_n}{\sigma^2}\right),\ P_{\mathrm{T}}=\frac{1}{N}\sum_{n=0}^{N-1}p_n
\end{align*}
where $p_n$ is the transmit power of the $n$-th symbol. Considering a rate constraint $R$, let us find the power allocation that minimizes the consumed power $P_{\mathrm{cons}}$. Under an ideal consumption model, we have $P_{\mathrm{cons}}^{\mathrm{ideal}}=P_{\mathrm{T}}$ and the problem can be written as\footnote{This problem can be seen as a conventional waterfilling problem where water/noise levels are the same at each time slot.}
\begin{align*}
	\min_{p_0,...,p_{N-1}} \frac{1}{N}\sum_{n=0}^{N-1}p_n \ \text{s.t.}\ \frac{1}{N}\sum_{n=0}^{N-1}\log_2\left(1+\frac{p_n}{\sigma^2}\right)=R.
\end{align*}
Given the concavity of the $\log(.)$ function, we can write using the Jensen's inequality
\begin{align*}
	R&\leq \log_2\left(1+\frac{1}{\sigma^2}\frac{1}{N}\sum_{n=0}^{N-1}p_n\right) \leftrightarrow
	(2^{R}-1)\sigma^2\leq \frac{1}{N}\sum_{n=0}^{N-1}p_n=P_{\mathrm{T}}
\end{align*}
and the bound is tight if uniform power allocation is used, \textit{i.e.}, $p_n=P_{\mathrm{T}}=(2^{R}-1)\sigma^2$, for $n=0,...,N-1$. 
Intuitively, the log dependence of the rate implies diminishing returns. Starting from a non-uniform allocation, it can always be improved by reallocating some power from the time interval with the highest allocated power to the one with the lowest power. 

In practice however, the consumed power~$P_{\mathrm{cons}}$ is far from being equal or even linearly proportional to the transmit power~$P_{\mathrm{T}}$. This is due to two main reasons, namely: i) as soon as a given time slot is active, a static load-independent power consumption is present due to activation of hardware components such as radio-frequency chains and baseband processing units; ii) the load-dependent power consumption, \textit{i.e.}, the dependence of $P_{\mathrm{cons}}$ in $p_n$, is typically concave as \glspl{pa} are more energy-efficient close to their saturation. Intuitively, this implies that the ``cost" of using more power decreases when a large output power is transmitted. These two effects will counterbalance the log penalty and push towards using a reduced number of active time slots, especially in low-to-medium load scenarios. 




\subsection{Contributions}

This paper presents a pioneering information-theoretic study of time-domain energy-saving techniques, using a realistic power consumption model. The transmission model considers single-antenna base stations and users. Even for such a basic system, a comprehensive study of energy-saving features is lacking, which is the gap this paper is aiming to fill. The investigated techniques provide drastic energy reduction by dictating how to dynamically (de)activate hardware resources as a function of the load. The optimization problems are formalized as the minimization of~$P_{\mathrm{cons}}$ for a given rate. More specifically, the structure of our paper and our contributions are structured as follows. Section~\ref{section:system_model} presents the hardware power consumption model used in this work, with two distinct contributions: active and sleep energy consumption. The active power consumption model is shown to address a large variety of \gls{pa} classes. Section~\ref{section:opt_allocation} considers the optimal allocation of time resources in a single-user scenario. The solution is approached step by step through lemmas to get more insight on its nature. Linear and exponential scaling regimes of~$P_{\mathrm{cons}}$ as a function of the load $R$ are identified. Asymptotic results for large $N$ are provided that greatly simplify the analysis while having negligible performance penalty. We prove that a rush-to-sleep approach is optimal in a noise limited regime but not otherwise. The optimal trade-off \gls{ee}-\gls{se} is also derived from previous results and we show that a maximal \gls{se} does not always provide a maximal \gls{ee}. Section~\ref{section:successive_sleep} then extends previous results by considering successive sleep modes, resulting in drastic energy reductions. Section~\ref{section:opt_allocation_MU} considers the extension to a multi-user scenario where users are multiplexed using~\gls{tdma}. The optimal allocation is provided for the most promising regime in terms of energy-savings, \textit{i.e.}, the low-to-medium-load scenario where~$P_{\mathrm{cons}}$ linearly scales with the rate of each user and the system is not fully active. Finally, Section~\ref{section:conclusion} concludes the paper.

{\textbf{Notations}}: 
The operators $\ceil{.}$, $\floor{.}$ and $\round{.}$ are the ceil, floor and round operators, respectively. The operator $\ceilfloor{x}$ which we refer to as the ceil-floor operator selects among the upper and lower bounding integers of $x$ the one that optimizes the cost function. The function $W(z)$ is the Lambert W function, \textit{i.e.}, the solution of $z=W(z)e^{W(z)}$.  We use the notation $f(x)=O(g(x))$, as $x\rightarrow a$, if there exist positive numbers $\delta$ and $\lambda$ such that $|f(x)|\leq \lambda g(x)$ when $0<|x-a|<\delta$.

\section{Power Consumption Model}
\label{section:system_model}

As described in the introduction, we consider the transmission of $N$ symbols, each of duration $T$~[s]. The full frame has thus a duration $NT$~[s]. Out of the $N$ intervals, a number $N_\mathrm{a}$ are active and actually transmitting information while the remaining $N-N_{\mathrm{a}}$ are inactive and in sleep mode. This implies that $0\leq N_{\mathrm{a}}\leq N$. For minimizing latency and maximizing the sleep duration which allows entering a deeper sleep mode~\cite{debaillie_flexible_2015}, the $N_\mathrm{a}$ active intervals are grouped together at the beginning of the transmission. We define as $E_{\mathrm{active}}$ the energy consumed during active time slots, which is assumed to depend on the transmit power at each time interval, \textit{i.e.}, $p_0,...,p_{N_{\mathrm{a}}-1}$. On the other hand, $E_{\mathrm{sleep}}$ represents the energy consumed during sleep modes, which is non zero as all hardware components cannot be switched-off. It is assumed to depend on $(N-N_\mathrm{a})T$ as a longer sleep duration allows to enter a deeper sleep mode~\cite{debaillie_flexible_2015}. The average consumed power over the frame duration is thus given by
\begin{align}
	P_{\mathrm{cons}}&=\frac{E_{\mathrm{active}}(p_0,...,p_{N_\mathrm{a}-1})+E_{\mathrm{sleep}}((N-N_\mathrm{a})T)}{NT}\label{eq:P_cons}.
\end{align}

In the following, we detail the models of the sleep and active energy consumption. As a benchmark, we also introduce the ideal consumption model
\begin{align}
	P_{\mathrm{cons}}^{\mathrm{ideal}}&=\frac{1}{N}\sum_{n=0}^{N_{\mathrm{a}}-1}p_n,\label{eq:ideal_cons_model}
\end{align}
implying that the consumed power is equal to the transmit power. In other words, no losses are present.


\subsection{Active Energy Consumption}

Using the well-established model from~\cite{auer_how_2011}, the active power consumption can be modelled as
\begin{align*}
	P_{\mathrm{active}}&=\frac{P_{\mathrm{PA}}+P_{\mathrm{RF}}+P_{\mathrm{BB}}}{(1-\sigma_{\mathrm{DC}})(1-\sigma_{\mathrm{MS}})(1-\sigma_{\mathrm{cool}})}
\end{align*}
where $P_{\mathrm{PA}}$, $P_{\mathrm{RF}}$ and $P_{\mathrm{BB}}$ are the powers consumed by the \glspl{pa}, the radio-frequency chains and the baseband unit respectively. The coefficients $\sigma_{\mathrm{DC}}$, $\sigma_{\mathrm{MS}}$ and $\sigma_{\mathrm{cool}}$ are the loss factors related to DC-DC power supply, mains supply and active cooling respectively.

\begin{table}[t!]
	\caption{Values of loss factors and efficiency used in evaluations~\cite{auer_how_2011}.}
		\centering{ \begin{tabular}{ |c|c|c|  }
				\hline
				DC-DC &$\sigma_{\mathrm{DC}}$ & 7.5\%\\
				\hline
				Mains supply &$\sigma_{\mathrm{MS}}$ & 9.0\%\\
				\hline
				Cooling &$\sigma_{\mathrm{cool}}$ & 10.0\%\\
				\hline\hline
				Efficiency &$\eta=(1-\sigma_{\mathrm{DC}})(1-\sigma_{\mathrm{MS}})(1-\sigma_{\mathrm{cool}})$ & 75.8\%\\
				\hline
			\end{tabular}
		}
	\label{table:loss_efficiency} 
     \vspace{-1em}
\end{table}

We are interested in modelling the dependence of the active consumed power $P_{\mathrm{active}}$ in the output power at each active time slot $p_0,...,p_{N_\mathrm{a}-1}$. The \gls{bs} power consumption analysis of~\cite{auer_how_2011} showed that mainly the \gls{pa} consumed power $P_{\mathrm{PA}}$ scales with the output power. The other terms are thus considered load-independent. 
The active energy consumed across the frame duration can thus be written as
\begin{align*}
	E_{\mathrm{active}}(p_0,...,p_{N_\mathrm{a}-1})&=\frac{T}{\eta}\left(N_{\mathrm{a}}\tilde{P}_{0} + \sum_{n=0}^{N_{\mathrm{a}-1}} P_{\mathrm{PA}}(p_n)\right)
\end{align*}
where $\tilde{P}_0=P_{\mathrm{BB}}+P_{\mathrm{RF}}$ and $\eta=(1-\sigma_{\mathrm{DC}})(1-\sigma_{\mathrm{MS}})(1-\sigma_{\mathrm{cool}})$. Values of loss factors and efficiencies used in evaluations are shown in Table~\ref{table:loss_efficiency}. To model the PA consumption, we use the following model
\begin{align}
	P_{\mathrm{PA}}(p)&=P_{\mathrm{PA},0}+\beta p^{\alpha},\ 0 \leq p\leq P_{\mathrm{max}} \label{eq:PA_model}
\end{align}
with $\alpha \in ]0,1]$ and $\beta\geq 0
$. The first term $P_{\mathrm{PA},0}$ represents the load-independent consumption while the second is load-dependent. This load dependency does not typically scale linearly with $p$. The fact that $\alpha \in ]0,1]$ implies concavity of $P_{\mathrm{PA}}(p)$. This concavity comes from the fact that a typical \gls{pa} efficiency is improved when moving closer to saturation~\cite{cripps_rf_2006}. The constant $P_{\mathrm{max}}$ denotes the maximal transmit power.In practice, $P_{\mathrm{max}}$ is (much) lower than the \gls{pa} saturation power, that we denote by $P_{\mathrm{sat}}$. The use of the so-called back-off $P_{\mathrm{max}}/P_{\mathrm{sat}}$ is required as recent technologies, \textit{e.g.}, \gls{ofdm}, have high \gls{papr}. A back-off (typically from -12~dB to -6~dB) prevents the \gls{pa} to enter the saturation region, which would otherwise create nonlinear distortion impacting the signal quality and creating out-of-band emissions. The authors of~\cite{persson_amplifier-aware_2013} have justified in details the use of a similar model as~(\ref{eq:PA_model}) through their own measurements and a literature review~\cite{grebennikov_rf_2015,mikami_efficiency_2007}. Model~(\ref{eq:PA_model}) is more general as it also includes a load-independent component, which is useful for particular \gls{pa} architectures.

\subsubsection{Ideal Power Amplifier}

\gls{pa} consumed power is linearly proportional to the output power giving $P_{\mathrm{PA},0}=0$, $\beta=\alpha=1$ and
\begin{align*}
	P_{\mathrm{PA}}^{\mathrm{ideal}}(p)=p.
\end{align*}

\subsubsection{Class A Power Amplifier}

\gls{pa} consumed power is independent of the load and has a maximal efficiency of $1/2$ giving $P_{\mathrm{PA},0}=2P_{\mathrm{sat}}$, $\beta=0$ and
\begin{align*}
	P_{\mathrm{PA}}^{\mathrm{A}}(p)=2P_{\mathrm{sat}}
\end{align*}
where $P_{\mathrm{sat}}$ is the saturation power of the \gls{pa}. 

\subsubsection{Class B Power Amplifier}

The \gls{pa} consumed power has a load dependence that scales with the square root of the output power giving $P_{\mathrm{PA},0}=0$, $\beta=\frac{4}{\pi} \sqrt{P_{\mathrm{sat}}}$ and $\alpha=1/2$
\begin{align*}
	P_{\mathrm{PA}}^{\mathrm{B}}(p)=\frac{4}{\pi} \sqrt{P_{\mathrm{sat}}} \sqrt{p}
\end{align*}
This model has much relevance for typical base stations working with a significant back-off from saturation~\cite{persson_amplifier-aware_2013}. Some authors sometimes call it the ``traditional" \gls{pa} model~\cite{hossain_energy_2018}. Therefore, by default, we will use it in following evaluations, with a 8 dB back-off.

\subsubsection{Envelope Tracking Power Amplifier}

According to the curve fitted model proposed in~\cite{hossain_impact_2011}, \gls{pa} consumption was shown to be modelled as
\begin{align*}
	P_{\mathrm{PA}}^{\mathrm{ET}}(p)\approx\frac{aP_{\mathrm{sat}}}{(1+a)\eta_{\mathrm{max}}}+\frac{1}{(1+a)\eta_{\mathrm{max}}}p,
\end{align*}
where $a=0.0082$. This model can again be seen as a special case of~(\ref{eq:PA_model}).

\subsubsection{Doherty Power Amplifier}

The $\ell$-way Doherty \gls{pa} consumed power is given by~\cite{jingon_joung_spectral_2014}
\begin{align*}
	P_{\mathrm{PA}}^{\mathrm{Doherty}}(p)=\frac{4 P_{\mathrm{sat}}}{\ell \pi}  \begin{cases}\sqrt{\xi} & 0<\xi \leq \frac{1}{\ell^2} \\ (\ell+1) \sqrt{\xi}-1 & \frac{1}{\ell^2}<\xi \leq 1\end{cases}
\end{align*}
where $\xi=p/P_{\mathrm{sat}}$. The class B model is obtained as a special case when $\ell=1$. Except in such particular cases, the model proposed in~(\ref{eq:PA_model}) cannot exactly represent such \glspl{pa} but can provide an approximation depending on the operating range of the Doherty amplifier.

\begin{remark}
	We previously described how model~(\ref{eq:PA_model}) can address typical theoretical \gls{pa} models. However, it can also be fitted for practical~\glspl{pa} based on measurements and/or datasheets. The PA was identified as the main load-dependent contribution. However, more generally, the model proposed in~(\ref{eq:PA_model}) can take into account other load-dependent terms.
\end{remark}

In the light of this remark, we formalize the underlying assumption about the load-dependent active energy consumption model used throughout this work.

$\mathbf{(As1)}$: The active energy consumed across the frame is
\begin{align}
	E_{\mathrm{active}}(p_0,...,p_{N_\mathrm{a}-1})&=T\left(N_{\mathrm{a}}P_0 + \gamma \sum_{n=0}^{N_{\mathrm{a}-1}} p_n^{\alpha}\right) \label{eq:E_active_model}
\end{align}
where $\gamma=\frac{\beta}{\eta}\geq 0$ and  $P_0
=\frac{P_{\mathrm{BB}}+P_{\mathrm{RF}}+P_{\mathrm{PA},0}}{\eta}\geq 0$, $\alpha \in ]0,1]$, $0 \leq p_n\leq P_{\mathrm{max}}$ for $n=0,\ldots,N_{\mathrm{a}}-1$, $0 \leq N_{\mathrm{a}}\leq N$ and $p_n=0$ for $n=N_{\mathrm{a}},\ldots,N-1$. 
The averaged consumed power is then
\begin{align}
	P_{\mathrm{cons}}&=\frac{N_{\mathrm{a}}}{N}P_0 + \frac{\gamma}{N} \sum_{n=0}^{N_{\mathrm{a}-1}}  p_n^{\alpha}+\frac{E_{\mathrm{sleep}}((N-N_\mathrm{a})T)}{NT}. \label{eq:P_cons_As1}
\end{align}

\subsection{Sleep Energy Consumption}
\label{subsection:sleep_energy_consumption}

In the ideal consumption model~(\ref{eq:ideal_cons_model}), the sleep energy consumption is exactly null. In practice, not all hardware can be switched off as always-on reference signals are required to allow users to access the network. Moreover, different hardware components have different activation/reactivation latencies. Hence, depending on the sleep duration, more or less components can be switched off. Therefore, power models have been proposed that consider successive sleep modes as a function of the sleep depth. We define the power consumption in sleep mode as $P_{\mathrm{sleep}}(t)$, where $t=0$ is used a reference for the system entering sleep and $t$ is the sleep duration so that $E_{\mathrm{sleep}}(0)=0$. Given the fact that increasing sleep duration allows to switch off more hardware components, we introduce the following assumption.

$\mathbf{(As2)}$: $P_{\mathrm{sleep}}(t)$ is monotonically non-increasing.

\begin{proposition}\label{prop:concavity_E_sleep}
	Under $\mathbf{(As2)}$, the sleep energy consumption $E_{\mathrm{sleep}}(t)$ is a concave function of the sleep duration $t$.
\end{proposition}
\begin{proof}
	Directly follows from $\mathbf{(As2)}$.
\end{proof}

\begin{remark}
	$\mathbf{(As2)}$ does not imply continuity of $P_{\mathrm{sleep}}(t)$, which can have jump discontinuities. Switching-off components might lead to a non-continuous drop of $P_{\mathrm{sleep}}(t)$, as shown in Fig.~\ref{fig:Fig_1_P_sleep}. 
\end{remark}

\begin{figure}[t!]
	\centering 
	\resizebox{\linewidthmod}{!}{%
		{
\begin{tikzpicture}

\definecolor{darkgray176}{RGB}{176,176,176}
\definecolor{lightgray204}{RGB}{204,204,204}

\begin{axis}[
legend cell align={left},
legend style={fill opacity=0.8, draw opacity=1, text opacity=1, draw=lightgray204},
log basis x={10},
tick align=outside,
tick pos=left,
x grid style={darkgray176},
xlabel={Time [ms]},
xmin=0.5, xmax=3569.99286,
xmode=log,
xtick style={color=black},
y grid style={darkgray176},
ylabel={\(\displaystyle P_{\mathrm{sleep}}(t)\) [W]},
ymin=0, ymax=70,
ytick style={color=black}
]
\addplot [semithick, black]
table {%
1e-06 110
0 110
};
\addlegendentry{Piecewise constant sleep $\mathbf{(As3)}$}
\addplot [semithick, black, forget plot]
table {%
0 50
6 50
};
\addplot [semithick, black, forget plot]
table {%
6 25
50 25
};
\addplot [semithick, black, forget plot]
table {%
50 1
1000 1
};
\addplot [semithick, black, forget plot]
table {%
1000 0.1
10000 0.1
};
\addplot [semithick, black, forget plot]
table {%
0 110
0 50
};
\addplot [semithick, black, forget plot]
table {%
6 50
6 25
};
\addplot [semithick, black, forget plot]
table {%
50 25
50 1
};
\addplot [semithick, black, forget plot]
table {%
1000 1
1000 0.1
};
\addplot [semithick, black, dotted, forget plot]
table {%
0 0
0 50
};
\addplot [semithick, black, dotted, forget plot]
table {%
6 0
6 25
};
\addplot [semithick, black, dotted, forget plot]
table {%
50 0
50 1
};
\addplot [semithick, black, dotted, forget plot]
table {%
1000 0
1000 0.1
};
\addplot [semithick, black, dashed]
table {%
1e-06 50
10000 50
};
\addlegendentry{Constant sleep  $\mathbf{(As4)}$}
\addplot [semithick, black, forget plot]
table {%
1e-06 50
6 50
};
\end{axis}

\end{tikzpicture}}
	} 
     \vspace{-1em}
	\caption{Two classical sleep power models are shown: constant or successive sleep modes based on values in Table~\ref{table:P_sleep}.}
	\label{fig:Fig_1_P_sleep} 
		\vspace{-1em}
\end{figure}
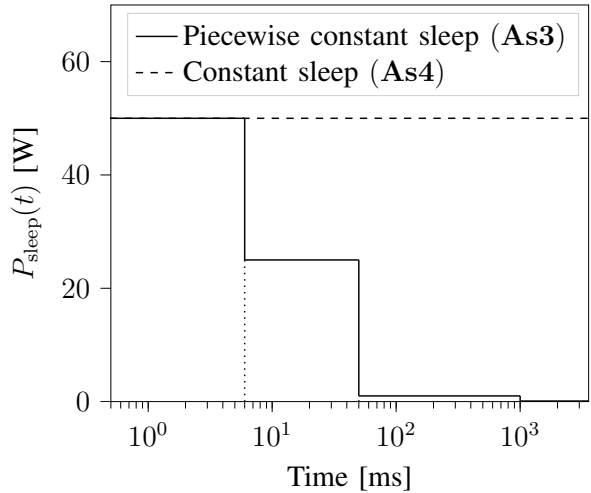
\begin{table}[t!]
		
		\caption{Numerical values of power and power models used in evaluations.}
		\centering{\resizebox{0.8\textwidth}{!}{%
			{\begin{tabular}{ |c|c|c|c|c|  }
					\hline
					Maximal transmit power& \multicolumn{4}{|c|}{$P_{\mathrm{max}}=20$ W}\\
					\hline\hline
					Deep sleep power & \multicolumn{4}{|c|}{$P_3=1 $ W}\\
					\hline \hline
					Load-independent active power & \multicolumn{4}{|c|}{$P_0=110 P_3$}\\
					\hline\hline
					\multirow{4}{14em}{Successive sleep power model $\mathbf{(As3)}$~\cite{3gpp.38.864}} & \multicolumn{4}{|c|}{$P_{\mathrm{sleep}}(t)$}\\ \cline{2-5}
					& Micro sleep    & Light sleep & Deep sleep & Hibernating sleep\\ \cline{2-5}
					&   $T_1=0$ & $T_2=6$ ms   & $T_3=50$ ms & $T_4=1$ s\\ \cline{2-5}
					&   $P_1 = 50 P_3$ & $P_2=25 P_3$  & $P_3$  & $P_4=0.1P_3$ \\
					\hline\hline
					Constant sleep power model $\mathbf{(As4)}$  & \multicolumn{4}{|c|}{$P_{\mathrm{sleep}}(t)=P_{\mathrm{sleep}}=P_1$}\\
					\hline
				\end{tabular}
			}
		} }

	\label{table:P_sleep} 
     \vspace{-1em}
\end{table}

Popular models for $P_{\mathrm{sleep}}(t)$ include the use of different sleep modes. The model in~\cite{debaillie_flexible_2015} has been used as a qualitative and quantitative reference for several years by companies~\cite{Ericsson_3GPP_document_for_discussion}. More recently, 3GPP has introduced an improved model, expressed in relative units with respect to the deep sleep mode, that better reflects current trends~\cite{3gpp.38.864,islam_enabling_2023}. The model contains four sleep modes and numerical values are given in Table~\ref{table:P_sleep} for one configuration described in~\cite{3gpp.38.864}. These values will be used as an example in evaluations.

We can formalize this model mathematically. Let us define as $S$ the number of sleep modes, starting from mode 0 (no sleep) to mode $S$ (deepest sleep mode). The start and end of each sleep mode are denoted by $T_s$ and $T_{s+1}$, with $T_0=0$ and $T_{S+1}=+\infty$. The sleep power consumption during mode $s$ is denoted by $P_s$. This notation is consistent with the definition of $P_0$ which denotes the load-independent active power consumption. This corresponds to ``sleep mode zero", taking place until $T_1$, where no hardware components are actually switched off.

$\mathbf{(As3)}$: the sleep power consumption is piecewise constant implying that $P_{\mathrm{sleep}}(t)=P_{s}$ and 
\begin{align*}
    E_{\mathrm{sleep}}(t)=\int_0^{t} P_{\mathrm{sleep}}(t') dt'=\sum_{s'=0}^{s-1}P_{s'}(T_{s'+1}-T_{s'})+(t-T_s)P_s
\end{align*}
where $s$ is the index such that $T_s < t \leq T_{s+1}$ and $P_{s+1} \leq P_{s}$ according to $\mathbf{(As2)}$.

We also introduce another popular sleep power model widely used in the literature and shown to be accurate to characterize 4G \gls{lte} macro base stations~\cite{auer_how_2011}.

$\mathbf{(As4)}$: the sleep power consumption is constant implying that $P_{\mathrm{sleep}}(t)=P_{\mathrm{sleep}}$ and $E_{\mathrm{sleep}}(t)=P_{\mathrm{sleep}}t$ with $P_0\geq P_{\mathrm{sleep}}$. 
The averaged consumed power is then
\begin{align}
	P_{\mathrm{cons}}&=\frac{N_{\mathrm{a}}}{N}P_0 + \frac{\gamma}{N} \sum_{n=0}^{N_{\mathrm{a}-1}}  p_n^{\alpha}+\frac{N-N_{\mathrm{a}}}{N}P_{\mathrm{sleep}}. \label{eq:P_cons_As4}
\end{align}

\begin{remark}
	$\mathbf{(As4)}$ is a particular case of $\mathbf{(As3)}$ with a single sleep mode: $S=1$, $P_{\mathrm{sleep}}=P_1$ and $T_1=0$. In the following, we use the term ``$\mathbf{(As3)}-\mathbf{(As4)}$", when both assumptions hold.
\end{remark}


\section{Optimal Allocation of Time Resources}
\label{section:opt_allocation}

This section considers the solution of minimizing the average consumed power under $(\mathbf{As1})$ and a rate constraint
\begin{align}
	\min_{N_{\mathrm{a}},p_0,...,p_{N_{\mathrm{a}}-1}} P_{\mathrm{cons}} \quad \text{s.t. } \frac{1}{N}\sum_{n=0}^{N_{\mathrm{a}}-1} \log_2\left(1+\frac{p_n}{\sigma^2}\right)=R. \label{eq:basic_problem}
\end{align}
We define a constant that will be useful throughout this section. Under $(\mathbf{As1})-(\mathbf{As4})$, $R_{\mathrm{a}}$ is the constant that minimizes the convex problem
\begin{align*}
	R_{\mathrm{a}}&= \arg \min_{x \geq 0} \frac{P_0-P_{\mathrm{sleep}} + \gamma\sigma^{2\alpha} \left(2^{x}-1\right)^{\alpha}}{x},
\end{align*}
where $P_{\mathrm{sleep}}$ is the constant sleep power consumption defined in $(\mathbf{As4})$. When $P_0-P_{\mathrm{sleep}}=0$, it is given by
\begin{align}
	R_{\mathrm{a}}=(W(-\alpha^{-1}e^{-\alpha^{-1}})+\alpha^{-1})/\log(2)\label{eq:R_a_delta_0}
\end{align}
where $W(z)$ is the Lambert W function, \textit{i.e.}, the solution of $z=W(z)e^{W(z)}$. In the following, the general solution of (\ref{eq:basic_problem}) is approached step by step, introducing several lemmas, which provide insight on its form and scaling as a function of $R$. 
Finally, the trade-off \gls{se} versus \gls{ee} will be characterized. 

\subsection{Optimal Allocation for Load-Dependent Consumed Power}

The following lemma provides the power allocation that minimizes the load-dependent part of the average consumed power under a rate constraint and without a maximal per-time slot power constraint. Under $\mathbf{(As1)}$, the load-dependent part of $P_{\mathrm{cons}}$ can be identified as
\begin{align}
	P_{\mathrm{ld}}(p_0,...,p_{N-1})&=\frac{\gamma}{N}  \sum_{n=0}^{N-1} p_n^{\alpha}, \label{eq:def_P_ld}
\end{align}
which equals $P_{\mathrm{cons}}$ for $P_0=0$ and $E_{\mathrm{sleep}}(t)=0$.

\begin{lemma} \label{lemma:PA_var_problem}
	Under $\mathbf{(As1)}$, for $P_0=E_{\mathrm{sleep}}(t)=0$, $\gamma>0$ and $P_{\mathrm{max}} \to +\infty$, the minimum of~(\ref{eq:basic_problem}), is achieved by uniformly allocating power among $N_{\mathrm{a}}$ time slots
	\begin{align*}
		N_{\mathrm{a}}&=\ceilfloor{\min(NR/R_{\mathrm{a}},N)},\\
  p_n&=\begin{cases}
			\left(2^{R\frac{N}{N_{\mathrm{a}}}}-1\right)\sigma^2 & \text{if } n=0,...,N_{\mathrm{a}}-1\\
			0 & \text{otherwise}
		\end{cases}.
	\end{align*}
\end{lemma}
\begin{proof}
	See Appendix~\ref{Appendix:proof_lemma_no_max_problem}.
\end{proof}

\begin{figure*} 
	\centering
	\subfloat[\label{fig:Fig_1_contour}]{%
		\resizebox{!}{0.45\linewidth}{\footnotesize
			{\input{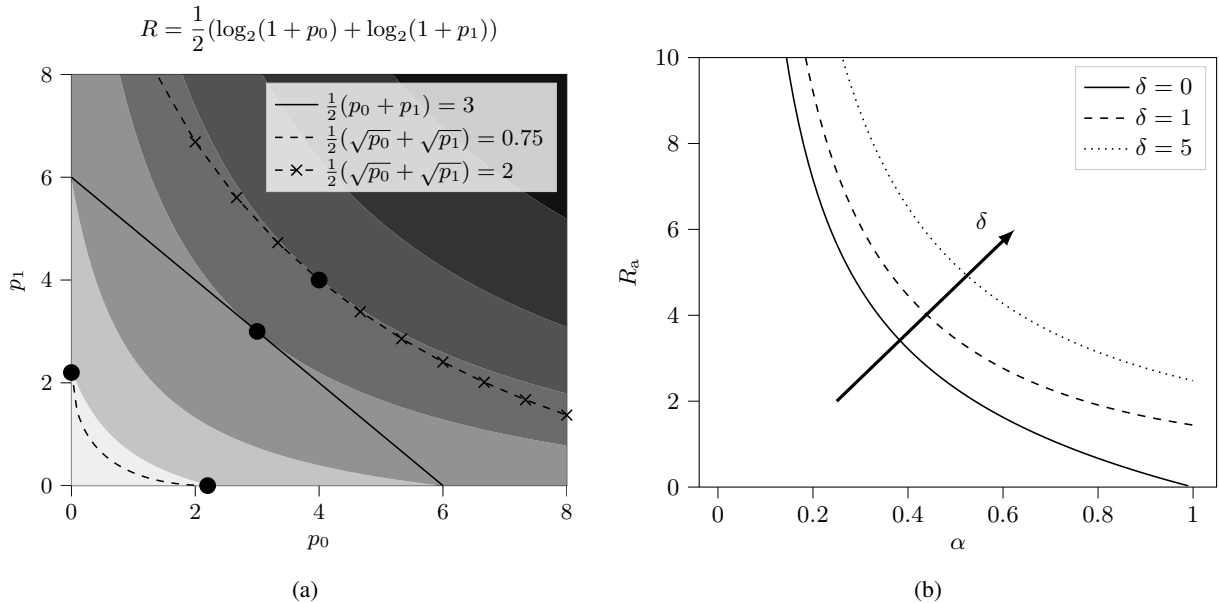}}
		}
	}
	\hfill
	\subfloat[\label{fig:Fig_2_R_a}]{%
		\resizebox{!}{0.415\linewidth}{\footnotesize
			{
\begin{tikzpicture}

\definecolor{darkgray176}{RGB}{176,176,176}
\definecolor{lightgray204}{RGB}{204,204,204}

\begin{axis}[
legend cell align={left},
legend style={fill opacity=0.8, draw opacity=1, text opacity=1, draw=lightgray204},
tick align=outside,
tick pos=left,
unbounded coords=jump,
x grid style={darkgray176},
xlabel={\(\displaystyle  \alpha \)},
xmin=-0.0395, xmax=1.0495,
xtick style={color=black},
y grid style={darkgray176},
ylabel={\(\displaystyle R_{\mathrm{a}}\)},
ymin=0, ymax=10,
ytick style={color=black}
]
\addplot [semithick, black]
table {%
0.01 144.269504088896
0.02 72.1347520444482
0.03 48.0898346962986
0.04 36.0673760217232
0.05 28.8539007583069
0.06 24.0449159588932
0.07 20.6099162768132
0.08 18.0336208026806
0.09 16.02970528602
0.1 14.4262951287901
0.11 13.1139300287674
0.12 12.0195630359614
0.13 11.0925721656575
0.14 10.2967722143865
0.15 9.60562159936604
0.16 8.99922350080642
0.17 8.46238145342835
0.18 7.98329943190178
0.19 7.55268958939662
0.2 7.16314567958468
0.21 6.80869422713798
0.22 6.48446737262088
0.23 6.1864607004691
0.24 5.91135148711072
0.25 5.6563605895441
0.26 5.41914630373463
0.27 5.19772194429389
0.28 4.99039123087928
0.29 4.7956971846189
0.3 4.61238137586201
0.31 4.43935117569567
0.32 4.27565324886861
0.33 4.12045195267297
0.34 3.97301162097873
0.35 3.8326819467359
0.36 3.6988858520174
0.37 3.57110936772329
0.38 3.44889314657489
0.39 3.33182531104651
0.4 3.21953539826889
0.41 3.11168921098887
0.42 3.00798442056054
0.43 2.90814679704303
0.44 2.8119269645658
0.45 2.71909759853772
0.46 2.62945099604338
0.47 2.5427969626704
0.48 2.4589609686492
0.49 2.37778253502641
0.5 2.29911381700011
0.51 2.22281835680208
0.52 2.1487699828446
0.53 2.07685183543294
0.54 2.00695550232053
0.55 1.93898024986398
0.56 1.87283233760946
0.57 1.80842440588279
0.58 1.74567492742124
0.59 1.68450771532328
0.6 1.62485148064108
0.61 1.56663943383231
0.62 1.50980892504694
0.63 1.45430111887409
0.64 1.40006069973015
0.65 1.34703560454687
0.66 1.29517677983017
0.67 1.24443796051508
0.68 1.1947754683501
0.69 1.1461480278107
0.7 1.09851659777366
0.71 1.05184421738569
0.72 1.00609586473634
0.73 0.961238327099285
0.74 0.917240081641267
0.75 0.874071185616838
0.76 0.831703175171401
0.77 0.790108971967198
0.78 0.749262796928234
0.79 0.709140090472007
0.8 0.669717438659725
0.81 0.630972504753241
0.82 0.592883965717218
0.83 0.555431453249829
0.84 0.518595498965239
0.85 0.482357483386739
0.86 0.446699588441338
0.87 0.411604753175128
0.88 0.377056632434383
0.89 0.343039558280362
0.9 0.309538503926402
0.91 0.276539050004611
0.92 0.244027352986165
0.93 0.211990115594398
0.94 0.180414559063598
0.95 0.149288397108718
0.96 0.118599811482595
0.97 0.0883374290071597
0.98 0.0584902999747739
0.99 0.0290478778237019
1 nan
};
\addlegendentry{$\delta=0$}
\addplot [semithick, black, dashed]
table {%
0.01 181.411164997928
0.02 91.7126690318153
0.03 61.3894465038366
0.04 46.049802464391
0.05 36.8867145831839
0.06 30.7332811286386
0.07 26.3486108789819
0.08 23.0510077678551
0.09 20.493186161812
0.1 18.4370625968087
0.11 16.7659767299658
0.12 15.3697308466314
0.13 14.1873311917073
0.14 13.1716280796744
0.15 12.2938683864302
0.16 11.5246288283914
0.17 10.8451397309396
0.18 10.2387568031152
0.19 9.6981377756534
0.2 9.2102909848893
0.21 8.76789906000501
0.22 8.36482139833238
0.23 7.99592375724577
0.24 7.65690084730727
0.25 7.34376721973047
0.26 7.05441224227373
0.27 6.78554985388435
0.28 6.53501932032367
0.29 6.30093760717065
0.3 6.0816704316624
0.31 5.87578939024099
0.32 5.6817132654894
0.33 5.49917415697093
0.34 5.32664395963781
0.35 5.16330530821773
0.36 5.00842154338956
0.37 4.86133082394476
0.38 4.72144098507957
0.39 4.58821745275965
0.4 4.46118006166242
0.41 4.3398955296986
0.42 4.22397303616381
0.43 4.1129575926745
0.44 4.00676601479673
0.45 3.90495105722692
0.46 3.80725059193011
0.47 3.71342401145688
0.48 3.62325013699537
0.49 3.53652451397952
0.5 3.45305968524924
0.51 3.37268172927654
0.52 3.29523007268909
0.53 3.22055631272944
0.54 3.14852311398059
0.55 3.07900245079518
0.56 3.01182958400206
0.57 2.94700269976379
0.58 2.88435033702639
0.59 2.8237779518591
0.6 2.7651966783515
0.61 2.70852293254805
0.62 2.65367806688544
0.63 2.60058815909864
0.64 2.54918329894957
0.65 2.49939767782182
0.66 2.45116888083275
0.67 2.40443855709278
0.68 2.3591506481492
0.69 2.31525267222049
0.7 2.27266505646193
0.71 2.23140931093737
0.72 2.19139710035205
0.73 2.15258658905302
0.74 2.11493728075219
0.75 2.07841051442349
0.76 2.04296929113359
0.77 2.00857785434794
0.78 1.97520213883691
0.79 1.94280922036212
0.8 1.91136777069001
0.81 1.88084729676177
0.82 1.85121876953207
0.83 1.82245399277565
0.84 1.79451322607522
0.85 1.76740137267707
0.86 1.74107303083704
0.87 1.71550470803944
0.88 1.69067345837513
0.89 1.6665569208391
0.9 1.64313352152857
0.91 1.62038219538386
0.92 1.59828255506567
0.93 1.57680395782983
0.94 1.55595641693911
0.95 1.53569821775729
0.96 1.51600859795662
0.97 1.49688796879673
0.98 1.47830430161135
0.99 1.46024716746454
1 1.44270069246741
};
\addlegendentry{$\delta=1$}
\addplot [semithick, black, dotted]
table {%
0.01 257.124958354918
0.02 130.136339436355
0.03 87.1311782919059
0.04 65.3570636487194
0.05 52.3544001880685
0.06 43.6201832884984
0.07 37.3971949092856
0.08 32.7161713660194
0.09 29.0863842315662
0.1 26.1783939701581
0.11 23.7959866745588
0.12 21.8148919837536
0.13 20.1371912558135
0.14 18.6988725756911
0.15 17.4511843439221
0.16 16.3611456524407
0.17 15.3988407159723
0.18 14.5432665472043
0.19 13.7776431144595
0.2 13.087833685597
0.21 12.464595536297
0.22 11.8976675520849
0.23 11.379802774814
0.24 10.9048880627761
0.25 10.4677677012402
0.26 10.0640628065022
0.27 9.68975377754972
0.28 9.34237825079534
0.29 9.01865120837089
0.3 8.71623093951977
0.31 8.43305888368438
0.32 8.16733001609233
0.33 7.91745295649252
0.34 7.6820265645208
0.35 7.45980886377789
0.36 7.24953645570998
0.37 7.05061848507863
0.38 6.86191461213328
0.39 6.68265127236548
0.4 6.51212990316219
0.41 6.34971619944051
0.42 6.19483802334726
0.43 6.04697556440904
0.44 5.90565579188187
0.45 5.77044865550457
0.46 5.64096119566823
0.47 5.51683516728817
0.48 5.3977414144774
0.49 5.28332472252474
0.5 5.17343605969512
0.51 5.06773854917582
0.52 4.96599974699956
0.53 4.86800394783653
0.54 4.77355027570808
0.55 4.68245363505433
0.56 4.59454078593368
0.57 4.50965107765149
0.58 4.42763492192228
0.59 4.34835284437389
0.6 4.27167398624085
0.61 4.1974767063807
0.62 4.12564669998593
0.63 4.05607722192932
0.64 3.98866800126688
0.65 3.92329779890173
0.66 3.85994011693468
0.67 3.79847366001954
0.68 3.7388205309015
0.69 3.68090741924683
0.7 3.62466439767927
0.71 3.57002558002129
0.72 3.51692855685546
0.73 3.46531436558284
0.74 3.41512668631487
0.75 3.36631223772835
0.76 3.31882061242578
0.77 3.27260370867577
0.78 3.227615721443
0.79 3.18381324051552
0.8 3.14115469315529
0.81 3.09960059055199
0.82 3.05911319988833
0.83 3.01965664170831
0.84 2.98119650366852
0.85 2.94368547005748
0.86 2.90712451825801
0.87 2.87146469951844
0.88 2.83667741759499
0.89 2.80273512006065
0.9 2.76961162995714
0.91 2.73728148785584
0.92 2.70572046155625
0.93 2.67490512507157
0.94 2.64481320678961
0.95 2.61542301953239
0.96 2.58671371528257
0.97 2.55866531136991
0.98 2.53125864474165
0.99 2.5044751266373
1 2.47829679056337
};
\addlegendentry{$\delta=5$}
\draw[-latex,very thick,draw=black] (axis cs:0.25,2) -- (axis cs:0.625,6);
\draw (axis cs:0.525,6) node[
  anchor=base west,
  text=black,
  rotate=0.0
]{$\delta$};
\end{axis}

\end{tikzpicture}}
	} }
	\caption{(a) Contour plot of i) rate constraint and ii) load-dependent power consumption ($\gamma=1$) for a frame of $N=2$ time slots and two load-dependent power exponents $\alpha$: $1$ and $1/2$. (b) Constant $R_{\mathrm{a}}$ as a function of the load-dependent power exponent $\alpha$ and the ratio $\delta=(P_0-P_{\mathrm{sleep}})/(\gamma\sigma^{2\alpha})$.}
\end{figure*}

%

\begin{remark}[Frame of $N=2$ symbols] \label{remark:N_2}
	To illustrate Lemma~\ref{lemma:PA_var_problem}, Fig.~\ref{fig:Fig_1_contour} considers the particular case $N=2$, $\gamma=1$. It shows contour plots of the constraint $R$ and the cost function $P_{\mathrm{ld}}$, as a function of $p_0$ and $p_1$. In the case $\alpha=1$, implying linearity of consumed and transmit power, the objective function curve $(p_0+p_1)/2=3$ is a straight line. As shown in the introduction, utilization of the two time slots is always optimal ($N_{\mathrm{a}}=2$), with uniform power allocation. However, for $\alpha=1/2$, this is not anymore the case. For low values of $R$, it is better to only use $N_{\mathrm{a}}=1$ slot while for large $R$, $N_{\mathrm{a}}=2$ is optimal, again with uniform power. 
\end{remark}

\begin{remark}[Scaling of $N_{\mathrm{a}}$ for general $N$ and $\alpha$]
	For an arbitrary value of $N$ and $\alpha \in ]0,1]$, the number of activated time slots $N_{\mathrm{a}}$ 
	scales approximately linearly with the rate $R$ up to the point where the maximal number $N$ is allocated, \textit{i.e.}, when $R > R_{\mathrm{a}}$. When $R \leq R_{\mathrm{a}}$, the rate per activated time slot is approximately equal to $R_{\mathrm{a}}$. The ``approximate" nature comes from the rounding operation. This error disappears for large $N$, as will be formalized properly in the following. The constant $R_{\mathrm{a}}$ is independent of $R$ and is a function of $\alpha$. As shown in Fig.~\ref{fig:Fig_2_R_a} for the case $\delta=0$ ($P_0=P_{\mathrm{sleep}}=0$), $R_{\mathrm{a}}$ monotonically decreases as a function of $\alpha$, implying that more time slots are activated for a fixed value of $R$. In the asymptotic cases of $\alpha$ approaching 0 or 1, a single ($N_{\mathrm{a}}=1$) or all time slots ($N_{\mathrm{a}}=N$) are allocated, respectively.
\end{remark}

%
%

\subsection{Optimal Allocation with no Maximal Power Constraint}

The following lemma gives the power allocation that minimizes the averaged consumed power under a rate constraint and without a maximal per-time slot power constraint. Removing this constraint provides the solution when it is not binding, \textit{i.e.}, when the user experiences a good channel (low normalized noise variance $\sigma^2=L\sigma_n^2$) and its target rate is not too high. The general case will be addressed in next subsection.

\begin{lemma} \label{lemma:no_max_problem}
	Under $\mathbf{(As1)}$-$\mathbf{(As2)}$ and for $P_{\mathrm{max}} \to +\infty$, the minimum of the problem (\ref{eq:basic_problem}) 	is achieved by uniformly allocating power among $N_{\mathrm{a}}$ time slots
	\begin{align}
		p_n&=\begin{cases}
			\left(2^{R\frac{N}{N_{\mathrm{a}}}}-1\right)\sigma^2 & \text{if }n=0,...,N_{\mathrm{a}}-1\\
			0 & \text{otherwise}
		\end{cases}\label{eq:opt_p_n}
	\end{align}
	where $N_{\mathrm{a}}$ is the argument that minimizes
	\begin{align}
		\min_{0 \leq N_{\mathrm{a}} \leq N} \frac{N_{\mathrm{a}}}{N}\left(P_0 + \gamma \sigma^{2\alpha} \left(2^{R\frac{N}{N_{\mathrm{a}}}}-1\right)^{\alpha}\right)+\frac{E_{\mathrm{sleep}}((N-N_\mathrm{a})T)}{NT}. \label{eq:theor_no_max_problem}
	\end{align}
	Under $\mathbf{(As3)}-\mathbf{(As4)}$, 
	the solution is
	\begin{align*}
		N_{\mathrm{a}}=\ceilfloor{\min(N R/R_{\mathrm{a}},N)}.
	\end{align*}
\end{lemma}
\begin{proof}
	See Appendix~\ref{Appendix:proof_lemma_no_max_problem}.
\end{proof}

\begin{remark}[Relation with Lemma~\ref{lemma:PA_var_problem}]
	Lemma~\ref{lemma:no_max_problem} extends the result of Lemma~\ref{lemma:PA_var_problem} by considering a non-zero sleep power $P_\mathrm{sleep}$ and load-independent active power consumption $P_0$. Under a constant sleep power model $\mathbf{(As4)}$, the problem has a similar solution.
\end{remark}

\begin{remark}[Scaling of $R_{\mathrm{a}}$]
	As shown in Fig.~\ref{fig:Fig_2_R_a}, $R_{\mathrm{a}}$ increases with $\delta=(P_0-P_{\mathrm{sleep}})/(\gamma\sigma^{2\alpha})$. This implies that less time slots should be activated if normalized noise power $\sigma^2$, $P_{\mathrm{sleep}}$ and $\gamma$ are small and if the active load-independent power consumption $P_0$ is high. A major difference with Lemma~\ref{lemma:PA_var_problem} is that $R_{\mathrm{a}}$ does not go to zero as $\alpha$ approaches 1 when $\delta>0$. This implies that, given nonzero static power consumption, not all time slots should always be activated. The particular load-independent case $\gamma=0$ implies that $R_{\mathrm{a}} = +\infty$ and $N_{\mathrm{a}}=1$. This makes sense as $P_{\mathrm{cons}}$ then only depends on $N_{\mathrm{a}}$. As shown in next subsection, considering a finite maximal transmit power per time slot will change this result.
\end{remark}

\subsection{Optimal Allocation with Maximal Power Constraint}

We now consider the additional constraint of a finite maximal per-time slot power $P_{\mathrm{max}}$. This constraint may render the problem unfeasible. Therefore, we introduce the following assumption.

\begin{algorithm}[t!]\linespread{1}
	\caption{Iterative resource allocation}\label{alg:max_constraint}
	\small 
	\begin{algorithmic}
		\Require $\sigma^2,R,N,P_{\mathrm{max}},R_{\mathrm{max}}, R_{\mathrm{a}}$
		\State $N_{\mathrm{a}} \gets (\ref{eq:theor_no_max_problem})$
		\Comment{Init. by sol. of Theor.~\ref{lemma:no_max_problem}}
		\State $p_0,...,p_{N-1} \gets (\ref{eq:opt_p_n})$ 
		\State $N_{\mathrm{max}} \gets 0$
		\State $\hat{N} \gets N$
		
		\While{$p_{N_{\mathrm{a}}-1}>P_{\mathrm{max}}$} \Comment{Check max power constraint}
			\State $N_{\mathrm{max}} \gets N_{\mathrm{max}}+1$ \Comment{Set one more time slot to max power}
			\State $R \gets (N R - R_{\mathrm{max}})/({N}-1) $ \Comment{Adapt rate constraint}
			\State $N \gets N-1$
			\State $N_{\mathrm{a}} \gets (\ref{eq:theor_no_max_problem})$
			\Comment{Update with sol. of Theor.~\ref{lemma:no_max_problem}}
			\State $p_0,...,p_{N-1} \gets (\ref{eq:opt_p_n})$ 		
		\EndWhile
		\State $p_{N_{\mathrm{a}}},...,p_{N_{\mathrm{a}}+N_{\mathrm{max}}-1} \gets P_{\mathrm{max}}$ 
		\State $p_{N_{\mathrm{a}}+N_{\mathrm{max}}},...,p_{\hat{N}-1} \gets 0$ 
	\end{algorithmic}
\end{algorithm}

$\textbf{(As5)}$: Problem~(\ref{eq:basic_problem}) is feasible, \textit{i.e.},
\begin{align*}
	R\leq  R_{\mathrm{max}}=\log_2\left(1+\frac{P_{\max}}{\sigma^2}\right).
\end{align*}

The exact solution of problem~(\ref{eq:basic_problem}) generally requires an iterative solution.

\begin{proposition}\label{prop:optimal}
	Under $\mathbf{(As1)}-\mathbf{(As2)},\mathbf{(As5)}$, the solution of the problem~(\ref{eq:basic_problem}) can be obtained by using Algorithm~\ref{alg:max_constraint}.
\end{proposition}
\begin{proof}
	See Appendix \ref{Appendix:proof_prop_optimal}.
\end{proof}

To avoid the need of an iterative solution and the ceil-floor operator, we use the fact that the problem greatly simplifies by considering the asymptotic case of a large $N$. Then, the ratio $N_{\mathrm{a}}/N$ can be considered asymptotically continuous instead of only taking discrete values.

\begin{remark}[Large $N$ assumption]
	The large $N$ assumption is realistic in practice as frames are typically made of many symbols. Moreover, the assumption of having a sufficiently large $N$ is central and implicit in this work. This article investigates the gain of activating only a share of transmission time slots. To be able to do this, some flexibility in the number of activated resources should be available, implying a sufficiently large $N$.
\end{remark}
To provide the closed-form asymptotic solution, we define
\begin{align*}
	\tilde{R}=\min(R_{\mathrm{a}},R_{\mathrm{max}}),P_{\mathrm{a}}=(2^{R_{\mathrm{a}}}-1)\sigma^2, \tilde{P}=\min(P_{\mathrm{a}},P_{\mathrm{max}}).
\end{align*}

\begin{theorem} \label{theor:main_SU_result}
	Under $\mathbf{(As1)}-\mathbf{(As5)}$, the solution of problem~(\ref{eq:basic_problem}) and two scaling regimes of $P_{\mathrm{cons}}$ as a function of $R$ can be found:

	$\square$ (\textbf{Linear}) If $R\leq \tilde{R}$, as $N\rightarrow+\infty$, the allocation
		\begin{align*}
			N_{\mathrm{a}}&=\round{NR/\tilde{R}},\ p_n=\begin{cases}
				\tilde{P} & \text{if }n=0,...,N_{\mathrm{a}}-1\\
				0 & \text{otherwise}
			\end{cases}
   		\end{align*}
        is asymptotically optimal and achieves a consumed power
        \begin{align*}
			P_{\mathrm{cons}}&=P_{\mathrm{sleep}}+R\frac{P_0-P_{\mathrm{sleep}} + \gamma \tilde{P}^{\alpha}}{\tilde{R}}+\epsilon
		\end{align*}
	    where $\epsilon$ is the gap from the optimum which asymptotically vanishes: $|\epsilon|=O(1/N)$.
     
	$\square$ (\textbf{Exponential}) If $R> \tilde{R}$:
	\begin{align*}
		N_{\mathrm{a}}&=N,\ p_n=\left(2^{R}-1\right)\sigma^2 \quad \text{for } n=0,...,N-1,\\
		P_{\mathrm{cons}}&=P_0 + \gamma \sigma^{2\alpha} \left(2^{R}-1\right)^{\alpha}.
	\end{align*}
				
\end{theorem}

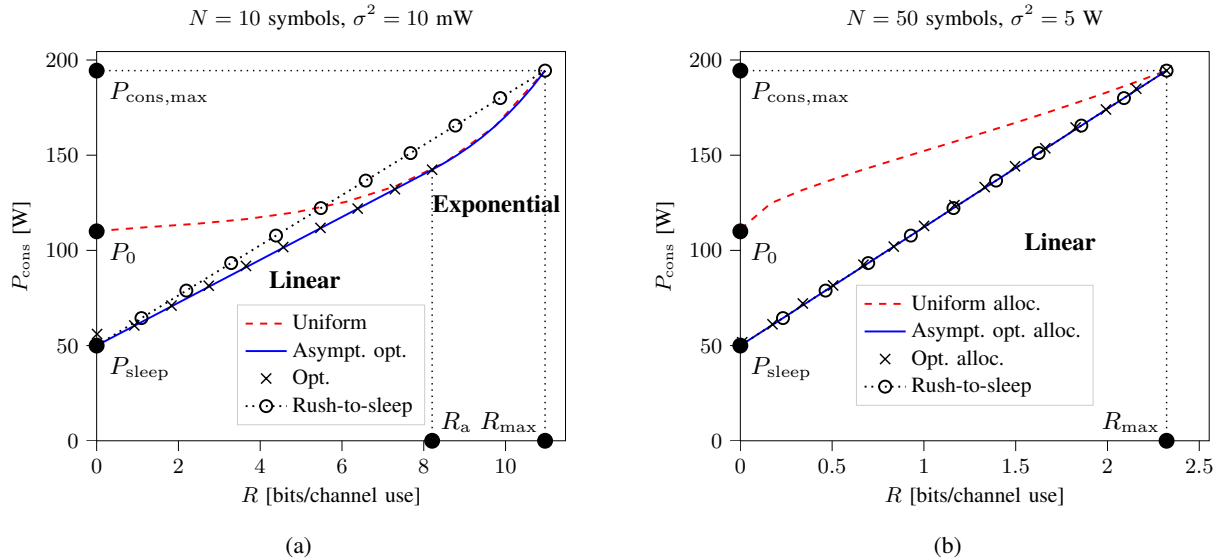
\begin{figure*} 
	\centering
	\subfloat[\label{fig:Fig_4a_scaling_regimes}]{%
		\resizebox{!}{0.415\linewidth}{\footnotesize
			{
\begin{tikzpicture}

\definecolor{darkgray176}{RGB}{176,176,176}
\definecolor{lightgray204}{RGB}{204,204,204}

\begin{axis}[
legend cell align={left},
legend style={
  fill opacity=0.8,
  draw opacity=1,
  text opacity=1,
  at={(0.5,0.04)},
  anchor=south,
  draw=lightgray204
},
tick align=outside,
tick pos=left,
title={\(\displaystyle N=10\) symbols, \(\displaystyle \sigma^2=10\) mW},
x grid style={darkgray176},
xlabel={\(\displaystyle  R \) [bits/channel use]},
xmin=0, xmax=11.4665054519057,
xtick style={color=black},
y grid style={darkgray176},
ylabel={\(\displaystyle P_{\mathrm{cons}}\) [W]},
ymin=0, ymax=204.433439236556,
ytick style={color=black}
]
\addplot [thick, red, dashed]
table {%
0.01 110.157458178149
1.22738949465619 112.186670675671
2.44477898931239 113.980220605102
3.66216848396858 116.4467778923
4.87955797862477 120.067954670843
6.09694747328097 125.505497271862
7.31433696793716 133.743692577241
8.53172646259335 146.271434658449
9.74911595724955 165.352235728105
10.9665054519057 194.433439236556
};
\addlegendentry{Uniform}
\addplot [thick, blue]
table {%
0.01 50.1125863469458
0.920539313920788 60.3640158574298
1.83107862784158 70.6154453679139
2.74161794176236 80.866874878398
3.65215725568315 91.1183043888821
4.56269656960394 101.369733899366
5.47323588352473 111.62116340985
6.38377519744551 121.872592920334
7.2943145113663 132.124022430818
8.20485382528709 142.375451941303
};
\addlegendentry{Asympt. opt.}
\addplot [thick, blue, forget plot]
table {%
8.20485382528709 142.375451941303
8.5117040060225 146.019927260467
8.8185541867579 150.072151189523
9.12540436749331 154.577996443482
9.43225454822871 159.588470575204
9.73910472896412 165.160293453497
10.0459549096995 171.35653938802
10.3528050904349 178.247351174815
10.6596552711703 185.910734149082
10.9665054519057 194.433439236556
};
\addplot [semithick, black, dotted, forget plot]
table {%
0 194.433439236556
10.9665054519057 194.433439236556
};
\addplot [semithick, black, dotted, forget plot]
table {%
10.9665054519057 0
10.9665054519057 194.433439236556
};
\addplot [semithick, black, dotted, forget plot]
table {%
8.20485382528709 0
8.20485382528709 142.375451941303
};
\addplot [semithick, black, mark=x, mark size=3, mark options={solid}, only marks]
table {%
0.01 56.050580303516
0.920539313920788 60.5833293069379
1.83107862784158 71.0089113541446
2.74161794176236 81.4353986361281
3.65215725568315 91.8304200108822
4.56269656960394 101.762406723826
5.47323588352473 111.815119361923
6.38377519744551 121.949414922617
7.2943145113663 132.141341038163
8.20485382528709 142.375451941303
};
\addlegendentry{Opt.}
\addplot [thick, black, dotted, mark=o, mark size=2.5, mark options={solid,fill opacity=0}]
table {%
0 50
1.09665054519057 64.4433439236556
2.19330109038115 78.8866878473112
3.28995163557172 93.3300317709668
4.3866021807623 107.773375694622
5.48325272595287 122.216719618278
6.57990327114344 136.660063541934
7.67655381633402 151.103407465589
8.77320436152459 165.546751389245
9.86985490671517 179.9900953129
10.9665054519057 194.433439236556
};
\addlegendentry{Rush-to-sleep}
\draw (axis cs:9.16650545190574,6) node[
  scale=1.14285714285714,
  anchor=base west,
  text=black,
  rotate=0.0
]{$R_{\mathrm{max}}$};
\draw (axis cs:8.20485382528709,6) node[
  scale=1.14285714285714,
  anchor=base west,
  text=black,
  rotate=0.0
]{$R_{\mathrm{a}}$};
\draw (axis cs:0.1,179.433439236556) node[
  scale=1.14285714285714,
  anchor=base west,
  text=black,
  rotate=0.0
]{$P_{\mathrm{cons,max}}$};
\draw (axis cs:0.1,95) node[
  scale=1.14285714285714,
  anchor=base west,
  text=black,
  rotate=0.0
]{$P_{0}$};
\draw (axis cs:0.1,35) node[
  scale=1.14285714285714,
  anchor=base west,
  text=black,
  rotate=0.0
]{$P_{\mathrm{sleep}}$};
\draw (axis cs:4,80) node[
  scale=1.14285714285714,
  anchor=base west,
  text=black,
  rotate=0.0
]{\bfseries Linear};
\draw (axis cs:8,120) node[
  scale=1.14285714285714,
  anchor=base west,
  text=black,
  rotate=0.0
]{\bfseries Exponential};
\addplot [semithick, black, mark=*, mark size=3, mark options={solid}, only marks, forget plot]
table {%
8.20485382528709 0
};
\addplot [semithick, black, mark=*, mark size=3, mark options={solid}, only marks, forget plot]
table {%
10.9665054519057 0
};
\addplot [semithick, black, mark=*, mark size=3, mark options={solid}, only marks, forget plot]
table {%
0 194.433439236556
};
\addplot [semithick, black, mark=*, mark size=3, mark options={solid}, only marks, forget plot]
table {%
0 50
};
\addplot [semithick, black, mark=*, mark size=3, mark options={solid}, only marks, forget plot]
table {%
0 110
};
\end{axis}

\end{tikzpicture}}
		} 
	}
	\hfill
	\subfloat[\label{fig:Fig_4b_scaling_regimes_only_linear}]{%
		\resizebox{!}{0.415\linewidth}{\footnotesize
			{
\begin{tikzpicture}

\definecolor{darkgray176}{RGB}{176,176,176}
\definecolor{lightgray204}{RGB}{204,204,204}

\begin{axis}[
legend cell align={left},
legend style={
  fill opacity=0.8,
  draw opacity=1,
  text opacity=1,
  at={(0.5,0.09)},
  anchor=south,
  draw=lightgray204
},
tick align=outside,
tick pos=left,
title={\(\displaystyle N=50\) symbols, \(\displaystyle \sigma^2=5\) W},
x grid style={darkgray176},
xlabel={\(\displaystyle  R \) [bits/channel use]},
xmin=0, xmax=2.5541209043761,
xtick style={color=black},
y grid style={darkgray176},
ylabel={\(\displaystyle P_{\mathrm{cons}}\) [W]},
ymin=0, ymax=204.433439236556,
ytick style={color=black}
]
\addplot [thick, red, dashed]
table {%
0.001 111.1116612733
0.16678057820624 124.778907954212
0.33256115641248 131.495188268546
0.498341734618721 137.117102948317
0.664122312824961 142.278649837764
0.829902891031201 147.226570957279
0.995683469237441 152.09040737885
1.16146404744368 156.950668994388
1.32724462564992 161.863542318928
1.49302520385616 166.871827669091
1.6588057820624 172.010461470869
1.82458636026864 177.309583353922
1.99036693847488 182.796370234451
2.15614751668112 188.496201551789
2.32192809488736 194.433439236556
};
\addlegendentry{Uniform alloc.}
\addplot [thick, blue]
table {%
0.01 50.622040964811
0.175137721063383 60.8942836985071
0.340275442126766 71.1665264322031
0.505413163190149 81.4387691658992
0.670550884253532 91.7110118995953
0.835688605316915 101.983254633291
1.0008263263803 112.255497366987
1.16596404744368 122.527740100683
1.33110176850706 132.79998283438
1.49623948957045 143.072225568076
1.66137721063383 153.344468301772
1.82651493169721 163.616711035468
1.9916526527606 173.888953769164
2.15679037382398 184.16119650286
2.32192809488736 194.433439236556
};
\addlegendentry{Asympt. opt. alloc.}
\addplot [thick, blue, forget plot]
table {%
2.32192809488736 194.433439236556
2.32192809488736 194.433439236556
2.32192809488736 194.433439236556
2.32192809488736 194.433439236556
2.32192809488736 194.433439236556
2.32192809488736 194.433439236556
2.32192809488736 194.433439236556
2.32192809488736 194.433439236556
2.32192809488736 194.433439236556
2.32192809488736 194.433439236556
2.32192809488736 194.433439236556
2.32192809488736 194.433439236556
2.32192809488736 194.433439236556
2.32192809488736 194.433439236556
2.32192809488736 194.433439236556
};
\addplot [semithick, black, dotted, forget plot]
table {%
0 194.433439236556
2.32192809488736 194.433439236556
};
\addplot [semithick, black, dotted, forget plot]
table {%
2.32192809488736 0
2.32192809488736 194.433439236556
};
\addplot [semithick, black, dotted, forget plot]
table {%
3.74713422369267 0
3.74713422369267 258.826022721954
};
\addplot [semithick, black, mark=x, mark size=3, mark options={solid}, only marks]
table {%
0.01 51.743408762457
0.175137721063383 61.1905059621076
0.340275442126766 72.1239740205917
0.505413163190149 81.5840464365394
0.670550884253532 92.4977251450539
0.835688605316915 101.98853104251
1.0008263263803 112.87164426702
1.16596404744368 123.783987721169
1.33110176850706 133.249670222416
1.49623948957045 144.187400627888
1.66137721063383 153.634547693781
1.82651493169721 164.567484831409
1.9916526527606 174.028483996107
2.15679037382398 184.94114088725
2.32192809488736 194.433439236556
};
\addlegendentry{Opt. alloc.}
\addplot [thick, black, dotted, mark=o, mark size=2.5, mark options={solid,fill opacity=0}]
table {%
0 50
0.232192809488736 64.4433439236556
0.464385618977472 78.8866878473112
0.696578428466209 93.3300317709668
0.928771237954945 107.773375694622
1.16096404744368 122.216719618278
1.39315685693242 136.660063541934
1.62534966642115 151.103407465589
1.85754247590989 165.546751389245
2.08973528539863 179.9900953129
2.32192809488736 194.433439236556
};
\addlegendentry{Rush-to-sleep}
\draw (axis cs:1.92192809488736,6) node[
  scale=1.14285714285714,
  anchor=base west,
  text=black,
  rotate=0.0
]{$R_{\mathrm{max}}$};
\draw (axis cs:3.14713422369267,6) node[
  scale=1.14285714285714,
  anchor=base west,
  text=black,
  rotate=0.0
]{$R_{\mathrm{a}}$};
\draw (axis cs:0.01,179.433439236556) node[
  scale=1.14285714285714,
  anchor=base west,
  text=black,
  rotate=0.0
]{$P_{\mathrm{cons,max}}$};
\draw (axis cs:0.01,95) node[
  scale=1.14285714285714,
  anchor=base west,
  text=black,
  rotate=0.0
]{$P_{0}$};
\draw (axis cs:0.01,35) node[
  scale=1.14285714285714,
  anchor=base west,
  text=black,
  rotate=0.0
]{$P_{\mathrm{sleep}}$};
\draw (axis cs:1.5,100) node[
  scale=1.14285714285714,
  anchor=base west,
  text=black,
  rotate=0.0
]{\bfseries Linear};
\addplot [semithick, black, mark=*, mark size=3, mark options={solid}, only marks, forget plot]
table {%
3.74713422369267 0
};
\addplot [semithick, black, mark=*, mark size=3, mark options={solid}, only marks, forget plot]
table {%
2.32192809488736 0
};
\addplot [semithick, black, mark=*, mark size=3, mark options={solid}, only marks, forget plot]
table {%
0 194.433439236556
};
\addplot [semithick, black, mark=*, mark size=3, mark options={solid}, only marks, forget plot]
table {%
0 50
};
\addplot [semithick, black, mark=*, mark size=3, mark options={solid}, only marks, forget plot]
table {%
0 110
};
\end{axis}

\end{tikzpicture}}
	} }\vspace{-0.5em}
	\caption{Power consumption versus load using optimal/uniform power allocation with constant sleep mode~$\mathbf{(As4)}$.}
	\label{Fig_4_5} 
    \vspace{-2em}
\end{figure*}
%
%
%

\begin{proof}
	See Appendix~\ref{Appendix:proof_theorem_main_SU_result}.
\end{proof}

\begin{remark}
	For large $N$, the solution has a simple form: the ceil-floor operator is replaced by a rounding operator. The power and number of bits per-activated time slot ($\tilde{P}$ and $\tilde{R}$) are constant in the linear regime. As shown in Fig.~\ref{fig:Fig_4a_scaling_regimes} and~\ref{fig:Fig_4b_scaling_regimes_only_linear}, the approximation error can barely be seen and is already negligible for small values of $N$ such as 10 or 50.
\end{remark}

\begin{remark}[Scaling regimes]
	Lemma~\ref{theor:main_SU_result} puts forward two scaling regimes: linear and exponential. They can easily be identified in Fig.~\ref{fig:Fig_4a_scaling_regimes}. In Fig.~\ref{fig:Fig_4b_scaling_regimes_only_linear}, a higher noise regime is considered so that only the linear regime is present ($\tilde{R}=R_{\mathrm{max}}$). 
\end{remark}

\begin{remark}[Gain with respect to uniform allocation]
	Fig.~\ref{fig:Fig_4a_scaling_regimes} and~\ref{fig:Fig_4b_scaling_regimes_only_linear} also plot the gain with respect to a uniform allocation, \textit{i.e.}, $p_n=\left(2^{R}-1\right)\sigma^2$ for $n=0,...,N-1$. As expected, the gain is larger at low load (rate) where the optimal allocation only activates few resources. 
\end{remark}

\begin{remark}[Rush-to-sleep]
	If $R_{\mathrm{max}}\leq R_{\mathrm{a}}$ (high noise regime, Fig.~\ref{fig:Fig_4b_scaling_regimes_only_linear}), a rush-to-sleep approach is optimal: active time slots transmitting at full power $P_{\mathrm{max}}$ and rate $R_{\mathrm{max}}$. This minimizes $N_{\mathrm{a}}$ and maximize sleep duration. If $R_{\mathrm{a}}<R_{\mathrm{max}}$ (low noise regime, Fig.~\ref{fig:Fig_4a_scaling_regimes}), reduced transmit power should be used instead and not even using sleep for $R>R_{\mathrm{a}}$. It is then better to fully activate the system with a uniform allocation. 
\end{remark}

\begin{remark}[Converse problem]
	We here consider the minimization of $P_{\mathrm{cons}}$ for a fixed $R$. The maximization of $R$ for a fixed $P_{\mathrm{cons}}$ can also be considered. It could occur if power is available and should directly be used, \textit{e.g.}, a solar panel or wind turbine without battery and/or not connected to the grid. The solution can be found using same methodology or directly by ``reverting" the result of Theor.~\ref{theor:main_SU_result}. Indeed, the allocation that minimizes $P_{\mathrm{cons}}$ for a fixed rate $R$ is also the allocation that maximizes $R$ for the minimal value of $P_{\mathrm{cons}}$ of the inital problem. Scaling regimes of $R$ as a function of $P_{\mathrm{cons}}$ will be linear and logarithmic. It is here omitted due to space constraints.
	
\end{remark}

\subsection{Trade-Off: Spectral Efficiency versus Energy Efficiency}

%

Considering that the transmission occupies a bandwidth $B=(1+\alpha_{\mathrm{rol}})/T$, where $\alpha_{\mathrm{rol}} \in [0,1]$ is the roll-off factor, so that the \gls{se} and the \gls{ee} are
\begin{align}
	\mathrm{SE}&=\frac{R}{TB}=\frac{R}{1+\alpha_{\mathrm{rol}}} \quad \text{[bits/s/Hz]},\ \mathrm{EE}=\frac{B \mathrm{SE}}{P_{\mathrm{cons}}}=\frac{R}{TP_{\mathrm{cons}}} \quad \text{[bits/Joule]}. \label{eq:definition_EE}
\end{align}
Given these relationships, the trade-off \gls{se}-\gls{ee} can be easily identified from previous results.
\begin{corollary}\label{corol:EE_SE_trade_off}
	Under $\mathbf{(As1)}-\mathbf{(As5)}$, the optimal \gls{ee} for a given \gls{se} is:
	
	$\square$ If $\mathrm{SE}\leq \tilde{R}/(1+\alpha_{\mathrm{rol}})$ (implying $R\leq \tilde{R}$), as $N\rightarrow+\infty$,
	\begin{align*}
		\mathrm{EE}&=\frac{B\mathrm{SE}}{P_{\mathrm{sleep}}+ \mathrm{SE}\frac{1+\alpha_{\mathrm{rol}}}{\tilde{R}}\left(P_0-P_{\mathrm{sleep}} + \gamma \tilde{P}^{\alpha}\right)}+O(1/N).
	\end{align*}

	$\square$ If $\mathrm{SE}>\tilde{R}/(1+\alpha_{\mathrm{rol}})$ (implying $R>\tilde{R}$):
	\begin{align*}
		\mathrm{EE}&=\frac{B \mathrm{SE}}{P_0 + \gamma \sigma^{2\alpha} \left(2^{\mathrm{SE}(1+\alpha_{\mathrm{rol}})}-1\right)^{\alpha}}.
	\end{align*}	
	
\end{corollary}
\begin{proof}
	From the definition of the \gls{ee} in (\ref{eq:definition_EE}), it is clear that its maximization is equivalent to the minimization of $P_{\mathrm{cons}}$. As a result, the optimal allocations and the results of Theor.~\ref{theor:main_SU_result} can directly be used, which lead to the above results. The two cases correspond to the linear and exponential scaling regimes of Theor.~\ref{theor:main_SU_result}, respectively.
\end{proof}

\begin{corollary}\label{corol:max_EE}
	Under $\mathbf{(As1)}-\mathbf{(As5)}$, the \gls{se} that maximizes the optimal \gls{ee} given in Corollary~\ref{corol:EE_SE_trade_off} is: 
	
	$\square$ If $R_{\mathrm{a}} \geq R_{\mathrm{max}}$:
	\begin{align*}
		\bar{\mathrm{SE}}=\mathrm{SE}_{\mathrm{max}}&=\frac{R_{\mathrm{max}}}{1+\alpha_{\mathrm{rol}}},\ 		\mathrm{EE}_{\mathrm{max}}=\frac{B\mathrm{SE}_{\mathrm{max}}}{P_0 + \gamma {P}_{\mathrm{max}}^{\alpha}}.
	\end{align*}
	
	$\square$ If $R_{\mathrm{a}} < R_{\mathrm{max}}$: 
	\begin{align*}
		\bar{\mathrm{SE}}&=\frac{\bar{R}}{1+\alpha_{\mathrm{rol}}},\ 		\mathrm{EE}_{\mathrm{max}}=\frac{B\bar{\mathrm{SE}}}{P_0 + \gamma \sigma^{2\alpha} \left(2^{\bar{R}}-1\right)^{\alpha}}
	\end{align*}
	where $\bar{R}$ is the constant that minimizes the convex problem
	\begin{align*}
		\min_{\bar{R}\in[R_{\mathrm{a}}, R_{\mathrm{max}}]} \frac{P_0 + \gamma \sigma^{2\alpha} \left(2^{\bar{R}}-1\right)^{\alpha}}{\bar{R}}.
	\end{align*}
\end{corollary}
\begin{proof}
	See Appendix~\ref{proof:corol_max_EE}.
\end{proof}

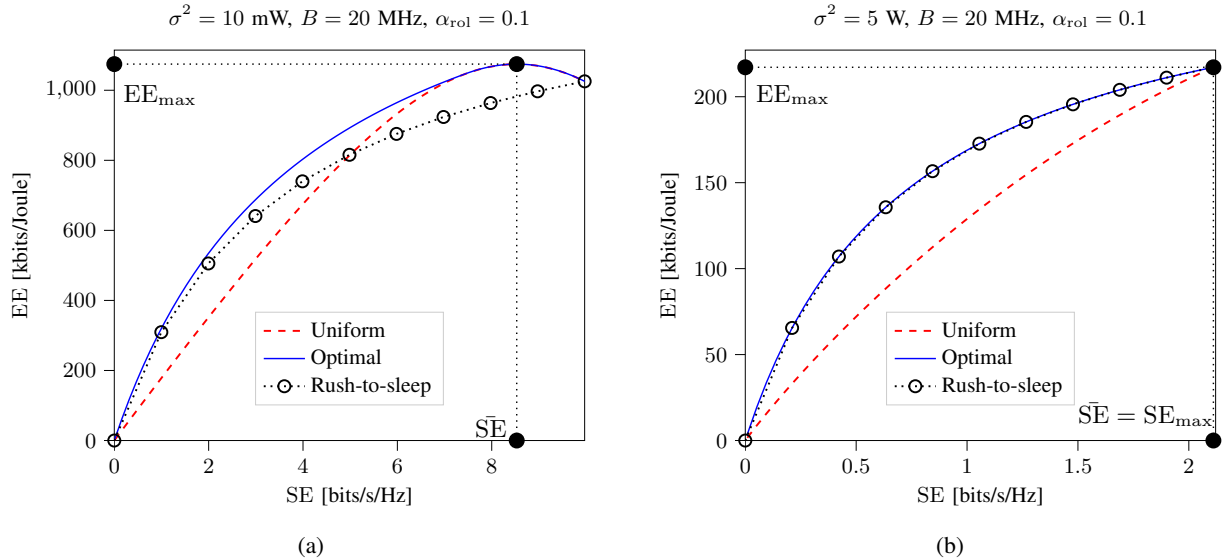
\begin{figure*} 
	\centering
	\subfloat[\label{fig:Fig_5a_SE_EE}]{%
		\resizebox{!}{0.415\linewidth}{\footnotesize
			{
\begin{tikzpicture}

\definecolor{darkgray176}{RGB}{176,176,176}
\definecolor{lightgray204}{RGB}{204,204,204}

\begin{axis}[
legend cell align={left},
legend style={
  fill opacity=0.8,
  draw opacity=1,
  text opacity=1,
  at={(0.5,0.09)},
  anchor=south,
  draw=lightgray204
},
tick align=outside,
tick pos=left,
title={\(\displaystyle \sigma^2=10\) mW, \(\displaystyle B=20\) MHz, \(\displaystyle \alpha_{\mathrm{rol}}=0.1\)},
x grid style={darkgray176},
xlabel={\(\displaystyle  \mathrm{SE} \) [bits/s/Hz]},
xmin=0, xmax=9.9695504108234,
xtick style={color=black},
y grid style={darkgray176},
ylabel={\(\displaystyle  \mathrm{EE} \) [kbits/Joule]},
ymin=0, ymax=1114.6596418829,
ytick style={color=black}
]
\addplot [thick, red, dashed]
table {%
0.00909090909090909 1.65052993074824
0.212365592799735 38.3360695969182
0.415640276508562 74.7871371431952
0.618914960217388 111.049874760477
0.822189643926215 147.131186599305
1.02546432763504 183.027577989607
1.22873901134387 218.730198330421
1.43201369505269 254.226551996757
1.63528837876152 289.501211407002
1.83856306247035 324.536133114945
2.04183774617917 359.310787099112
2.245112429888 393.802187764898
2.44838711359683 427.984869257435
2.65166179730565 461.830827969246
2.85493648101448 495.309445878728
3.0582111647233 528.387403766091
3.26148584843213 561.028591025137
3.46476053214096 593.194017640444
3.66803521584978 624.841733404174
3.87130989955861 655.926759330991
4.07458458326744 686.401036340264
4.27785926697626 716.213396518316
4.48113395068509 745.309562586502
4.68440863439391 773.632181534073
4.88768331810274 801.12089868506
5.09095800181157 827.712478714025
5.29423268552039 853.340980263864
5.49750736922922 877.937990805008
5.70078205293805 901.432928162476
5.90405673664687 923.753414677246
6.1073314203557 944.82572921359
6.31060610406453 964.575341130095
6.51388078777335 982.927528861199
6.71715547148218 999.808083881196
6.920430155191 1015.14409853259
7.12370483889983 1028.86483350483
7.32697952260866 1040.90265768348
7.53025420631749 1051.19404971919
7.73352889002631 1059.68064709119
7.93680357373514 1066.31032479507
8.14007825744396 1071.03828223906
8.34335294115279 1073.8281136851
8.54662762486162 1074.65283483871
8.74990230857044 1073.4958362043
8.95317699227927 1070.35173279981
9.15645167598809 1065.22707996193
9.35972635969692 1058.14092642078
9.56300104340575 1049.12517866547
9.76627572711458 1038.22475486394
9.9695504108234 1025.49751215315
};
\addlegendentry{Uniform}
\addplot [semithick, blue]
table {%
0.00909090909090909 3.62819393434207
0.161129013456161 61.9780586193923
0.313167117821413 116.249560323552
0.465205222186665 166.855846272782
0.617243326551918 214.156082844475
0.76928143091717 258.46399296245
0.921319535282422 300.054822883886
1.07335763964767 339.171065448477
1.22539574401293 376.027191268239
1.37743384837818 410.813582847204
1.52947195274343 443.699824016322
1.68151005710868 474.837464658875
1.83354816147393 504.362355845602
1.98558626583919 532.396631289677
2.13762437020444 559.05039607934
2.28966247456969 584.423171927347
2.44170057893494 608.605138932723
2.5937386833002 631.67820651413
2.74577678766545 653.716940317827
2.8978148920307 674.789367202456
3.04985299639595 694.957676609617
3.2018911007612 714.278833553068
3.35392920512646 732.805115952567
3.50596730949171 750.584586986483
3.65800541385696 767.661511450086
3.81004351822221 784.076723713559
3.96208162258746 799.867953719054
4.11411972695272 815.070116495325
4.26615783131797 829.715569865919
4.41819593568322 843.834344354308
4.57023404004847 857.454348723645
4.72227214441372 870.601554111446
4.87431024877898 883.300159315523
5.02634835314423 895.572739444394
5.17838645750948 907.440379853319
5.33042456187473 918.922797037655
5.48246266623999 930.038447941648
5.63450077060524 940.804628957399
5.78653887497049 951.237565730927
5.93857697933574 961.352494756054
6.09061508370099 971.163737619081
6.24265318806625 980.684768655089
6.3946912924315 989.928276688029
6.54672939679675 998.90622144948
6.698767501162 1007.62988520359
6.85080560552725 1016.10992004675
7.00284370989251 1024.35639129906
7.15488181425776 1032.37881735898
7.30691991862301 1040.18620635331
7.45895802298826 1047.78708987886
};
\addlegendentry{Optimal}
\addplot [semithick, blue, forget plot]
table {%
7.45895802298826 1047.78708987886
7.51019460233184 1050.25782905527
7.56143118167541 1052.61422473429
7.61266776101899 1054.85539078085
7.66390434036256 1056.98046690638
7.71514091970614 1058.98861977959
7.76637749904971 1060.87904413709
7.81761407839328 1062.65096389213
7.86885065773686 1064.30363324002
7.92008723708043 1065.83633775834
7.971323816424 1067.2483955002
8.02256039576758 1068.53915807891
8.07379697511115 1069.70801174206
8.12503355445473 1070.75437843322
8.1762701337983 1071.67771683943
8.22750671314188 1072.4775234224
8.27874329248545 1073.15333343163
8.32997987182902 1073.70472189743
8.3812164511726 1074.13130460181
8.43245303051617 1074.43273902532
8.48368960985975 1074.60872526788
8.53492618920332 1074.65900694144
8.5861627685469 1074.58337203274
8.63739934789047 1074.3816537339
8.68863592723404 1074.05373123914
8.73987250657762 1073.59953050556
8.79110908592119 1073.01902497604
8.84234566526477 1072.3122362625
8.89358224460834 1071.47923478754
8.94481882395192 1070.52014038284
8.99605540329549 1069.43512284235
9.04729198263906 1068.22440242886
9.09852856198264 1066.88825033203
9.14976514132621 1065.42698907662
9.20100172066979 1063.84099287917
9.25223830001336 1062.13068795195
9.30347487935694 1060.29655275272
9.35471145870051 1058.33911817917
9.40594803804408 1056.25896770681
9.45718461738766 1054.05673746944
9.50842119673123 1051.73311628108
9.55965777607481 1049.28884559871
9.61089435541838 1046.72471942509
9.66213093476195 1044.04158415098
9.71336751410553 1041.24033833642
9.7646040934491 1038.3219324307
9.81584067279268 1035.28736843079
9.86707725213625 1032.13769947813
9.91831383147982 1028.87402939395
9.9695504108234 1025.49751215315
};
\addplot [thick, black, dotted, mark=o, mark size=2.5, mark options={solid,fill opacity=0}]
table {%
0 0
0.99695504108234 309.405124061659
1.99391008216468 505.512434752993
2.99086512324702 640.922341178807
3.98782016432936 740.038091713656
4.9847752054117 815.727213261942
5.98173024649404 875.417454296525
6.97868528757638 923.696613415636
7.97564032865872 963.551415141437
8.97259536974106 997.009902588565
9.9695504108234 1025.49751215315
};
\addlegendentry{Rush-to-sleep}
\addplot [semithick, black, dotted, forget plot]
table {%
0 1074.6596418829
8.52979045339817 1074.6596418829
};
\addplot [semithick, black, dotted, forget plot]
table {%
8.52979045339817 0
8.52979045339817 1074.6596418829
};
\draw (axis cs:0.01,967.193677694614) node[
  scale=1.14285714285714,
  anchor=base west,
  text=black,
  rotate=0.0
]{$\mathrm{EE}_{\mathrm{max}}$};
\draw (axis cs:7.4635666467234,10) node[
  scale=1.14285714285714,
  anchor=base west,
  text=black,
  rotate=0.0
]{$\bar{\mathrm{SE}}$};
\addplot [semithick, black, mark=*, mark size=3, mark options={solid}, only marks, forget plot]
table {%
8.52979045339817 1074.6596418829
};
\addplot [semithick, black, mark=*, mark size=3, mark options={solid}, only marks, forget plot]
table {%
0 1074.6596418829
};
\addplot [semithick, black, mark=*, mark size=3, mark options={solid}, only marks, forget plot]
table {%
8.52979045339817 0
};
\end{axis}

\end{tikzpicture}}
		} }
	\hfill
	\subfloat[\label{fig:Fig_5b_SE_EE_single_regime}]{%
		\resizebox{!}{0.415\linewidth}{\footnotesize
			{
\begin{tikzpicture}

\definecolor{darkgray176}{RGB}{176,176,176}
\definecolor{lightgray204}{RGB}{204,204,204}

\begin{axis}[
legend cell align={left},
legend style={
  fill opacity=0.8,
  draw opacity=1,
  text opacity=1,
  at={(0.5,0.09)},
  anchor=south,
  draw=lightgray204
},
tick align=outside,
tick pos=left,
title={\(\displaystyle \sigma^2=5\) W, \(\displaystyle B=20\) MHz, \(\displaystyle \alpha_{\mathrm{rol}}=0.1\)},
x grid style={darkgray176},
xlabel={\(\displaystyle  \mathrm{SE} \) [bits/s/Hz]},
xmin=0, xmax=2.12084372262487,
xtick style={color=black},
y grid style={darkgray176},
ylabel={\(\displaystyle  \mathrm{EE} \) [kbits/Joule]},
ymin=0, ymax=227.12764336352,
ytick style={color=black}
]
\addplot [thick, red, dashed]
table {%
0.00909090909090909 1.60162777800967
0.0519838236528268 8.77447193122314
0.0948767382147444 15.60946104441
0.137769652776662 22.2119619023095
0.18066256733858 28.625731616191
0.223555481900497 34.8762292952778
0.266448396462415 40.9804238493972
0.309341311024333 46.9505360178294
0.35223422558625 52.7958241550225
0.395127140148168 58.5235584631483
0.438020054710086 64.1396029806855
0.480912969272003 69.6487870387446
0.523805883833921 75.0551546569259
0.566698798395839 80.3621388393462
0.609591712957756 85.5726874244567
0.652484627519674 90.6893564505022
0.695377542081592 95.7143810260267
0.738270456643509 100.649730191821
0.781163371205427 105.497150120145
0.824056285767345 110.258198642526
0.866949200329262 114.934273213974
0.90984211489118 119.526633829966
0.952735029453098 124.036422007061
0.995627944015015 128.464676654436
1.03852085857693 132.812347461469
1.08141377313885 137.080306280038
1.12430668770077 141.269356872501
1.16719960226269 145.380243315984
1.2100925168246 149.413657292958
1.25298543138652 153.370244451808
1.29587834594844 157.250609985308
1.33877126051036 161.055323547117
1.38166417507227 164.784923604472
1.42455708963419 168.439921307918
1.46745000419611 172.020803945024
1.51034291875803 175.528038033878
1.55323583331994 178.962072103091
1.59612874788186 182.323339197687
1.63902166244378 185.612259144152
1.6819145770057 188.829240602994
1.72480749156762 191.974682932944
1.76770040612953 195.048977887542
1.81059332069145 198.052511161929
1.85348623525337 200.985663805236
1.89637914981529 203.848813511931
1.9392720643772 206.642335803722
1.98216497893912 209.36660511215
2.02505789350104 212.021995770735
2.06795080806296 214.608882924473
2.11084372262487 217.12764336352
};
\addlegendentry{Uniform}
\addplot [semithick, blue]
table {%
0.00909090909090909 3.59168019212361
0.0519838236528268 19.4125343976966
0.0948767382147444 33.5895051737831
0.137769652776662 46.366153104435
0.18066256733858 57.9401880753927
0.223555481900497 68.4737736603736
0.266448396462415 78.1011720988442
0.309341311024333 86.9344975882804
0.35223422558625 95.0681023570912
0.395127140148168 102.58196013578
0.438020054710086 109.544304592942
0.480912969272003 116.013707358697
0.523805883833921 122.040729759892
0.566698798395839 127.669246918312
0.609591712957756 132.937517606966
0.652484627519674 137.879055051892
0.695377542081592 142.523340588031
0.738270456643509 146.896412288333
0.781163371205427 151.021353395399
0.824056285767345 154.918699905134
0.866949200329262 158.606782496069
0.90984211489118 162.102014820122
0.952735029453098 165.419137721338
0.995627944015015 168.571427047685
1.03852085857693 171.570871234314
1.08141377313885 174.428323666797
1.12430668770077 177.153633906348
1.16719960226269 179.755761121065
1.2100925168246 182.242872475934
1.25298543138652 184.622428758224
1.29587834594844 186.901259129381
1.33877126051036 189.085626580999
1.38166417507227 191.181285416155
1.42455708963419 193.193531867051
1.46745000419611 195.12724878647
1.51034291875803 196.986945207019
1.55323583331994 198.776791442839
1.59612874788186 200.500650308962
1.63902166244378 202.162104950219
1.6819145770057 203.764483701632
1.72480749156762 205.310882343268
1.76770040612953 206.804184062673
1.81059332069145 208.247077395727
1.85348623525337 209.642072380811
1.89637914981529 210.991515130471
1.9392720643772 212.297600998542
1.98216497893912 213.562386498182
2.02505789350104 214.787800106911
2.06795080806296 215.975652078043
2.11084372262487 217.12764336352
};
\addlegendentry{Optimal}
\addplot [thick, black, dotted, mark=o, mark size=2.5, mark options={solid,fill opacity=0}]
table {%
0 0
0.211084372262487 65.5100618343312
0.422168744524975 107.031682035403
0.633253116787462 135.701896757407
0.84433748904995 156.687583293743
1.05542186131244 172.71317125985
1.26650623357492 185.351330996023
1.47759060583741 195.573432872307
1.6886749780999 204.011853319836
1.89975935036239 211.095988038652
2.11084372262487 217.12764336352
};
\addlegendentry{Rush-to-sleep}
\addplot [semithick, black, dotted, forget plot]
table {%
0 217.12764336352
2.11084372262487 217.12764336352
};
\addplot [semithick, black, dotted, forget plot]
table {%
2.11084372262487 0
2.11084372262487 217.12764336352
};
\addplot [semithick, blue, forget plot]
table {%
2.11084372262487 217.12764336352
2.11084372262487 217.12764336352
2.11084372262487 217.12764336352
2.11084372262487 217.12764336352
2.11084372262487 217.12764336352
2.11084372262487 217.12764336352
2.11084372262487 217.12764336352
2.11084372262487 217.12764336352
2.11084372262487 217.12764336352
2.11084372262487 217.12764336352
2.11084372262487 217.12764336352
2.11084372262487 217.12764336352
2.11084372262487 217.12764336352
2.11084372262487 217.12764336352
2.11084372262487 217.12764336352
2.11084372262487 217.12764336352
2.11084372262487 217.12764336352
2.11084372262487 217.12764336352
2.11084372262487 217.12764336352
2.11084372262487 217.12764336352
2.11084372262487 217.12764336352
2.11084372262487 217.12764336352
2.11084372262487 217.12764336352
2.11084372262487 217.12764336352
2.11084372262487 217.12764336352
2.11084372262487 217.12764336352
2.11084372262487 217.12764336352
2.11084372262487 217.12764336352
2.11084372262487 217.12764336352
2.11084372262487 217.12764336352
2.11084372262487 217.12764336352
2.11084372262487 217.12764336352
2.11084372262487 217.12764336352
2.11084372262487 217.12764336352
2.11084372262487 217.12764336352
2.11084372262487 217.12764336352
2.11084372262487 217.12764336352
2.11084372262487 217.12764336352
2.11084372262487 217.12764336352
2.11084372262487 217.12764336352
2.11084372262487 217.12764336352
2.11084372262487 217.12764336352
2.11084372262487 217.12764336352
2.11084372262487 217.12764336352
2.11084372262487 217.12764336352
2.11084372262487 217.12764336352
2.11084372262487 217.12764336352
2.11084372262487 217.12764336352
2.11084372262487 217.12764336352
2.11084372262487 217.12764336352
};
\draw (axis cs:0.01,197.12764336352) node[
  scale=1.14285714285714,
  anchor=base west,
  text=black,
  rotate=0.0
]{$\mathrm{EE}_{\mathrm{max}}$};
\draw (axis cs:1.46084372262487,10) node[
  scale=1.14285714285714,
  anchor=base west,
  text=black,
  rotate=0.0
]{$\bar{\mathrm{SE}}=\mathrm{SE}_{\mathrm{max}}$};
\addplot [semithick, black, mark=*, mark size=3, mark options={solid}, only marks, forget plot]
table {%
2.11084372262487 217.12764336352
};
\addplot [semithick, black, mark=*, mark size=3, mark options={solid}, only marks, forget plot]
table {%
0 217.12764336352
};
\addplot [semithick, black, mark=*, mark size=3, mark options={solid}, only marks, forget plot]
table {%
2.11084372262487 0
};
\end{axis}

\end{tikzpicture}}
		} }
	\caption{Energy efficiency versus spectral efficiency for optimal/uniform allocation with constant sleep mode~$\mathbf{(As4)}$.}
	\label{fig5} 
    \vspace{-1em}
\end{figure*}

\begin{remark}[Scaling of EE-SE]
	Fig.~\ref{fig:Fig_5a_SE_EE} and \ref{fig:Fig_5b_SE_EE_single_regime} plot the EE as a function of the SE for the optimal and uniform allocation. Most gain in terms of EE (not absolute energy) is obtained at medium SE. The plots are obtained by varying the rate $R$, all other parameters being fixed.
\end{remark}

\begin{remark}[Optimal EE]
	As shown in Fig.~\ref{fig:Fig_5a_SE_EE}, for a low normalized noise
	, the maximal EE is not obtained at maximal SE while it is for a higher noise
	, as shown in Fig.~\ref{fig:Fig_5b_SE_EE_single_regime}.
\end{remark}

\section{Optimal Allocation for Piecewise Constant Successive Sleep Modes}
\label{section:successive_sleep}

Fig.~\ref{fig:Fig_4a_scaling_regimes} and \ref{fig:Fig_4b_scaling_regimes_only_linear} have shown promising gains to reduce consumed power at low load. However, they are still limited by the relatively high sleep power consumption $P_{\mathrm{sleep}}$. Similarly, Fig.~\ref{fig:Fig_5a_SE_EE} and~\ref{fig:Fig_5b_SE_EE_single_regime} can seem limited as one would hope to obtain an EE that is approximately flat as a function of the SE. The reason is the same: a too high $P_{\mathrm{sleep}}$. 

To drastically reduce energy consumption at low-to-medium load, it is of paramount importance to implement successive sleep modes and use a frame duration long enough so that the system can enter these sleep modes. 
To find the optimal allocation, the iterative algorithm of Prop.~\ref{prop:optimal} can be used, which requires to solve problem~(\ref{eq:theor_no_max_problem}) at each iteration. This is an integer programming problem which can have a significant complexity. If the problem is relaxed by considering $N_{\mathrm{a}}$ continuous, it remains challenging to solve as it implies the minimization of the concave function $E_{\mathrm{sleep}}(t)$ (Prop.~\ref{prop:concavity_E_sleep}).

If the sleep power consumption is assumed piecewise constant, according to $\mathbf{(As3)}$
, a simple allocation can still be found. This sleep power model was detailed in Section~\ref{subsection:sleep_energy_consumption} and Fig.~\ref{fig:Fig_1_P_sleep}. For a given $N_{\mathrm{a}}$, only sleep modes such that $T_s\leq (N-N_{\mathrm{a}})T$ 
can be entered. We define $N_{\mathrm{a},s}^+=N-\floor{T_{s}/T}$ so that sleep mode $s$ can be used if $N_{\mathrm{a}}\leq N_{\mathrm{a},s}^+$. Depending on the target rate $R$, it might be unfeasible to use a given sleep mode because of the maximal power constraint per time slot. A minimum of active time slots $ \frac{NR}{R_\mathrm{max}}$ is required to satisfy it. Otherwise, the power per time slot would have to be higher than $P_{\mathrm{max}}$ to satisfy the rate constraint. Mode $s$ is thus feasible only if $\frac{NR}{R_\mathrm{max}}\leq N_{\mathrm{a},s}^+$. As a result, deepest sleep modes are only possible for low values of $R$, which intuitively makes sense. For a target rate $R$, we define the set of feasible sleep modes as $\mathcal{S}_R=\left\lbrace s \ | \ s \in \{0,\cdots,S-1\}  \text{ and } \frac{NR}{R_\mathrm{max}}\leq N_{\mathrm{a},s}^+ \right\rbrace $. Moreover, we generalize the definition of $R_{\mathrm{a}}$ per sleep mode as the constant $R_{\mathrm{a}}(s)$ that minimizes the convex problem
\begin{align*}
	R_{\mathrm{a}}(s)= \arg \min_{x \geq 0} \frac{P_0- P_{s} +\gamma \sigma^{2\alpha}\left(2^{x}-1\right)^{\alpha}}{x}.
\end{align*}
We also define $\tilde{R}_s=\min(R_{\mathrm{a}}(s),R_{\mathrm{max}})$. 

\begin{theorem} \label{theor:successive_sleep_mode}
	Under $\mathbf{(As1)}-\mathbf{(As3)}, \mathbf{(As5)}$, as $N\rightarrow +\infty$, the solution of problem~(\ref{eq:basic_problem}) is found by computing for all feasible sleep modes $s \in \mathcal{S}_R$
	\begin{align*}
		N_{\mathrm{a},s}&=\round{\min\left(NR/\tilde{R}_s,N_{\mathrm{a},s}^+\right)},\ p_{n,s}=\begin{cases}
			(2^{\frac{RN}{N_{\mathrm{a},s}}}-1)\sigma^2 & \text{if }n=0,...,N_{\mathrm{a},s}-1\\
			0 & \text{otherwise}
		\end{cases}\\
		P_{\mathrm{cons},s}&=\frac{N_{\mathrm{a},s}}{N}\left(P_0 + \gamma \sigma^{2\alpha} \left(2^{\frac{RN}{N_{\mathrm{a},s}}}-1\right)^{\alpha}\right) + \frac{E_{\mathrm{sleep}}((N-N_\mathrm{a,s})T)}{NT}+ O({1}/{N})
	\end{align*}
	and choosing among these modes the one that has the minimal $P_{\mathrm{cons},s}$.
	
\end{theorem}

\begin{proof}
	See Appendix~\ref{Appendix:proof_theorem_successive_sleep_mode}.
\end{proof}

\begin{figure*} 
	\centering
	\subfloat[\label{fig:Fig_6a_scaling_regimes_P_sleep_decreasing}]{%
		\resizebox{!}{0.415\linewidth}{\footnotesize
			{
\begin{tikzpicture}

\definecolor{darkgray176}{RGB}{176,176,176}
\definecolor{lightgray204}{RGB}{204,204,204}

\begin{axis}[
legend cell align={left},
legend style={
  fill opacity=0.8,
  draw opacity=1,
  text opacity=1,
  at={(0.6,0.04)},
  anchor=south,
  draw=lightgray204
},
tick align=outside,
tick pos=left,
title={\(\displaystyle \sigma^2=10\) mW, \(\displaystyle NT=200\) ms},
x grid style={darkgray176},
xlabel={\(\displaystyle  R \) [bits/channel use]},
xmin=0, xmax=11.514830724501,
xtick style={color=black},
y grid style={darkgray176},
ylabel={\(\displaystyle P_{\mathrm{cons}}\) [W]},
ymin=0, ymax=204.433439236556,
ytick style={color=black}
]
\addplot [semithick, red, dashed]
table {%
0.01 110.157458178149
0.58665818167925 111.337362872287
1.1633163633585 112.102135644248
1.73997454503775 112.888248621515
2.316632726717 113.767321891808
2.89329090839625 114.787321203048
3.4699490900755 115.994323677541
4.04660727175475 117.439091787158
4.623265453434 119.180699857867
5.19992363511325 121.289543619031
5.7765818167925 123.850449775184
6.35323999847174 126.966214174822
6.92989818015099 130.761791110897
7.50655636183024 135.389336977828
8.08321454350949 141.034326999945
8.65987272518874 147.922998339548
9.23653090686799 156.331422040495
9.81318908854724 166.596569545132
10.3898472702265 179.12981843192
10.9665054519057 194.433439236556
};
\addlegendentry{Uniform}
\addplot [semithick, blue]
table {%
0.01 7.91841260101241
0.58665818167925 17.6133015062002
1.1633163633585 27.3079580016899
1.73997454503775 37.0027654759946
2.316632726717 46.6975056241276
2.89329090839625 56.3922962232647
3.4699490900755 66.0870532530065
4.04660727175475 75.7818366033257
4.623265453434 85.4766008835062
5.19992363511325 95.1713802024471
5.7765818167925 104.866148514656
6.35323999847174 114.560925266597
6.92989818015099 124.030057271859
7.50655636183024 132.208239696246
8.08321454350949 140.386422366376
8.65987272518874 147.922998339548
9.23653090686799 156.331422040495
9.81318908854724 166.596569545132
10.3898472702265 179.12981843192
10.9665054519057 194.433439236556
};
\addlegendentry{Asympt. opt.}
\addplot [semithick, black, mark=x, mark size=3, mark options={solid}, only marks]
table {%
0.01 7.91841260101241
0.58665818167925 17.6133015062002
1.1633163633585 27.3079580016899
1.73997454503775 37.0027654759946
2.316632726717 46.6975056241276
2.89329090839625 56.3922962232647
3.4699490900755 66.0870532530065
4.04660727175475 75.7818366033257
4.623265453434 85.4766008835062
5.19992363511325 95.1713802024471
5.7765818167925 104.866148514656
6.35323999847174 114.560925266597
6.92989818015099 124.030057271859
7.50655636183024 132.208239696246
8.08321454350949 140.386422366376
8.65987272518874 147.922998339548
9.23653090686799 156.331422040495
9.81318908854724 166.596569545132
10.3898472702265 179.12981843192
10.9665054519057 194.433439236556
};
\addlegendentry{Opt.}
\addplot [thick, black, dotted, mark=o, mark size=2.5, mark options={solid,fill opacity=0}]
table {%
0 7.75029026248129
0.783042160078953 21.5620317288544
1.56608432015791 35.3737731952276
2.34912648023686 49.1855146616007
3.13216864031581 62.9972561279739
3.91521080039477 76.808997594347
4.69825296047372 90.6207390607201
5.48129512055267 104.432480527093
6.26433728063163 118.244221993466
7.04737944071058 132.05596345984
7.83042160078953 145.867704926213
8.61346376086849 158.828777813837
9.39650592094744 170.926845592242
10.1795480810264 183.024913370646
10.9625902411053 194.381874281455
};
\addlegendentry{Rush-to-sleep}
\addplot [semithick, black, dotted, forget plot]
table {%
0 194.433439236556
10.9665054519057 194.433439236556
};
\addplot [semithick, black, dotted, forget plot]
table {%
10.9665054519057 0
10.9665054519057 194.433439236556
};
\draw (axis cs:9.16650545190574,6) node[
  scale=1.14285714285714,
  anchor=base west,
  text=black,
  rotate=0.0
]{$R_{\mathrm{max}}$};
\draw (axis cs:0.2,179.433439236556) node[
  scale=1.14285714285714,
  anchor=base west,
  text=black,
  rotate=0.0
]{$P_{\mathrm{cons,max}}$};
\addplot [semithick, black, mark=*, mark size=3, mark options={solid}, only marks, forget plot]
table {%
10.9665054519057 0
};
\addplot [semithick, black, mark=*, mark size=3, mark options={solid}, only marks, forget plot]
table {%
0 110
};
\end{axis}

\end{tikzpicture}}
		} 
	}
	\hfill
	\subfloat[\label{fig:Fig_6b_scaling_regimes_P_sleep_decreasing_single_regime}]{%
		\resizebox{!}{0.415\linewidth}{\footnotesize
			{
\begin{tikzpicture}

\definecolor{darkgray176}{RGB}{176,176,176}
\definecolor{lightgray204}{RGB}{204,204,204}

\begin{axis}[
legend cell align={left},
legend style={
  fill opacity=0.8,
  draw opacity=1,
  text opacity=1,
  at={(0.55,0.04)},
  anchor=south,
  draw=lightgray204
},
tick align=outside,
tick pos=left,
title={\(\displaystyle \sigma^2=5\) W, \(\displaystyle NT=200\) ms},
x grid style={darkgray176},
xlabel={\(\displaystyle  R \) [bits/channel use]},
xmin=0, xmax=2.5541209043761,
xtick style={color=black},
y grid style={darkgray176},
ylabel={\(\displaystyle P_{\mathrm{cons}}\) [W]},
ymin=0, ymax=204.433439236556,
ytick style={color=black}
]
\addplot [semithick, red, dashed]
table {%
0.0001 110.351483336586
0.122301478678282 122.55689288766
0.244502957356564 128.142610019831
0.366704436034847 132.711028226617
0.488905914713129 136.812775599801
0.611107393391411 140.659663949167
0.733308872069693 144.36056142258
0.855510350747975 147.980942623553
0.977711829426258 151.56448800236
1.09991330810454 155.142646881536
1.22211478678282 158.739460525711
1.3443162654611 162.374231966754
1.46651774413939 166.063111874932
1.58871922281767 169.820094573939
1.71092070149595 173.657672796216
1.83312218017423 177.587284985735
1.95532365885252 181.619631239646
2.0775251375308 185.764903184624
2.19972661620908 190.032955816812
2.32192809488736 194.433439236556
};
\addlegendentry{Uniform}
\addplot [semithick, blue]
table {%
0.0001 7.77493855589247
0.122301478678282 17.9375211258036
0.244502957356564 28.1210941115847
0.366704436034847 38.3046673982527
0.488905914713129 48.4882407601227
0.611107393391411 58.6524450282831
0.733308872069693 68.836017668657
0.855510350747975 79.0195905368236
0.977711829426258 89.2031635473055
1.09991330810454 99.3867366526368
1.22211478678282 109.570309824348
1.3443162654611 119.734514231498
1.46651774413939 129.918087078382
1.58871922281767 140.101660016133
1.71092070149595 150.285233025264
1.83312218017423 159.519601275683
1.95532365885252 168.439340043415
2.0775251375308 177.343994312294
2.19972661620908 186.263732872075
2.32192809488736 194.433439236556
};
\addlegendentry{Asympt. opt.}
\addplot [semithick, black, mark=x, mark size=3, mark options={solid}, only marks]
table {%
0.0001 7.77493855589247
0.122301478678282 17.9569028109231
0.244502957356564 28.1210941115847
0.366704436034847 38.3046673982527
0.488905914713129 48.4882407601227
0.611107393391411 58.6718141520704
0.733308872069693 68.8553875590561
0.855510350747975 79.0389609746346
0.977711829426258 89.2031635473055
1.09991330810454 99.3867366526368
1.22211478678282 109.570309824348
1.3443162654611 119.753883044327
1.46651774413939 129.937456300502
1.58871922281767 140.121029584516
1.71092070149595 150.285233025264
1.83312218017423 159.519601275683
1.95532365885252 168.439340043415
2.0775251375308 177.359078849341
2.19972661620908 186.278817687091
2.32192809488736 194.433439236556
};
\addlegendentry{Opt.}
\addplot [thick, black, dotted, mark=o, mark size=2.5, mark options={solid,fill opacity=0}]
table {%
0 7.75029026248129
0.0828963975325727 14.6561609956679
0.165792795065145 21.5620317288544
0.248689192597718 28.467902462041
0.331585590130291 35.3737731952276
0.414481987662864 42.2796439284141
0.497378385195436 49.1855146616007
0.580274782728009 56.0913853947873
0.663171180260582 62.9972561279739
0.746067577793155 69.9031268611604
0.828963975325727 76.808997594347
0.9118603728583 83.7148683275336
0.994756770390873 90.6207390607201
1.07765316792345 97.5266097939067
1.16054956545602 104.432480527093
1.24344596298859 111.33835126028
1.32634236052116 118.244221993466
1.40923875805374 125.150092726653
1.49213515558631 132.05596345984
1.57503155311888 138.961834193026
1.65792795065145 145.867704926213
1.74082434818403 152.773575659399
1.8237207457166 158.828777813837
1.90661714324917 164.877811703039
1.98951354078175 170.926845592242
2.07240993831432 176.975879481444
2.15530633584689 183.024913370646
2.23820273337946 189.073947259848
2.32109913091204 194.381874281455
};
\addlegendentry{Rush-to-sleep}
\addplot [semithick, black, dotted, forget plot]
table {%
0 194.433439236556
2.32192809488736 194.433439236556
};
\addplot [semithick, black, dotted, forget plot]
table {%
2.32192809488736 0
2.32192809488736 194.433439236556
};
\draw (axis cs:1.92192809488736,6) node[
  scale=1.14285714285714,
  anchor=base west,
  text=black,
  rotate=0.0
]{$R_{\mathrm{max}}$};
\draw (axis cs:0.2,179.433439236556) node[
  scale=1.14285714285714,
  anchor=base west,
  text=black,
  rotate=0.0
]{$P_{\mathrm{cons,max}}$};
\addplot [semithick, black, mark=*, mark size=3, mark options={solid}, only marks, forget plot]
table {%
2.32192809488736 0
};
\addplot [semithick, black, mark=*, mark size=3, mark options={solid}, only marks, forget plot]
table {%
0 110
};
\end{axis}

\end{tikzpicture}}
	} }
	\caption{Power consumption versus load with successive sleep power model~$\mathbf{(As3)}$.}
	\label{Fig_6} 
    \vspace{-2em}
\end{figure*}
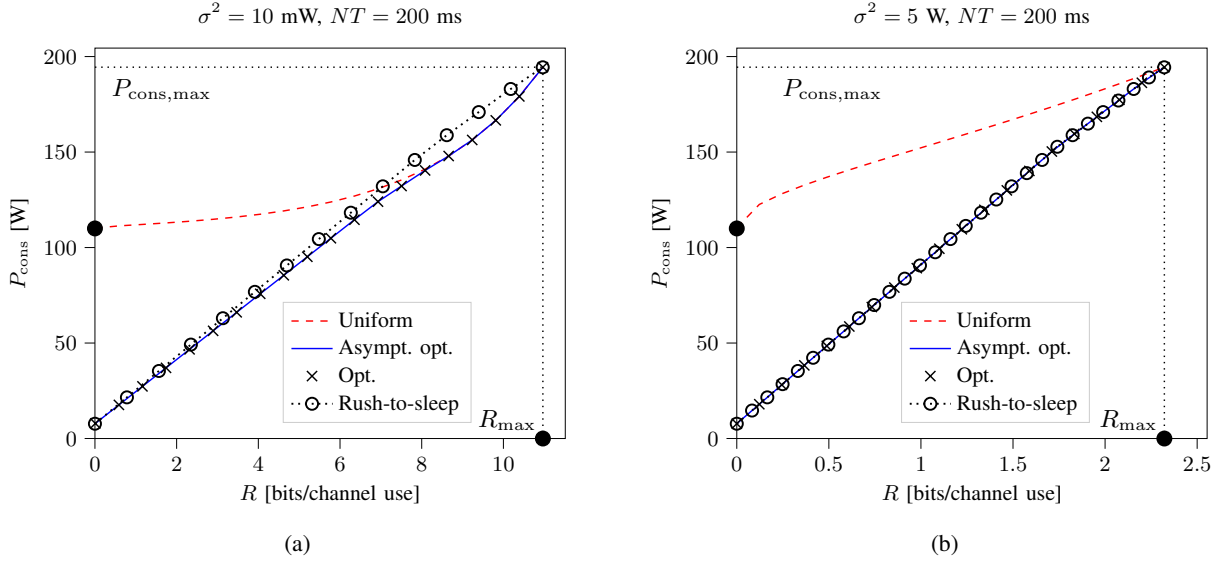

\begin{remark}
    Fig.~\ref{fig:Fig_6a_scaling_regimes_P_sleep_decreasing} and Fig.~\ref{fig:Fig_6b_scaling_regimes_P_sleep_decreasing_single_regime} are plotted based on the same simulation parameters as Fig.~\ref{fig:Fig_4a_scaling_regimes} and~\ref{fig:Fig_4b_scaling_regimes_only_linear}, but with a different sleep model. Using successive power modes has a drastic impact on the consumed energy at low load. A key parameter to allow drastic savings is to have a large enough frame duration so that deepest sleep modes can be used. For these figures, it was fixed to 200 ms. Hence, deep sleep power mode can be entered but not hibernating sleep power mode. 
\end{remark}

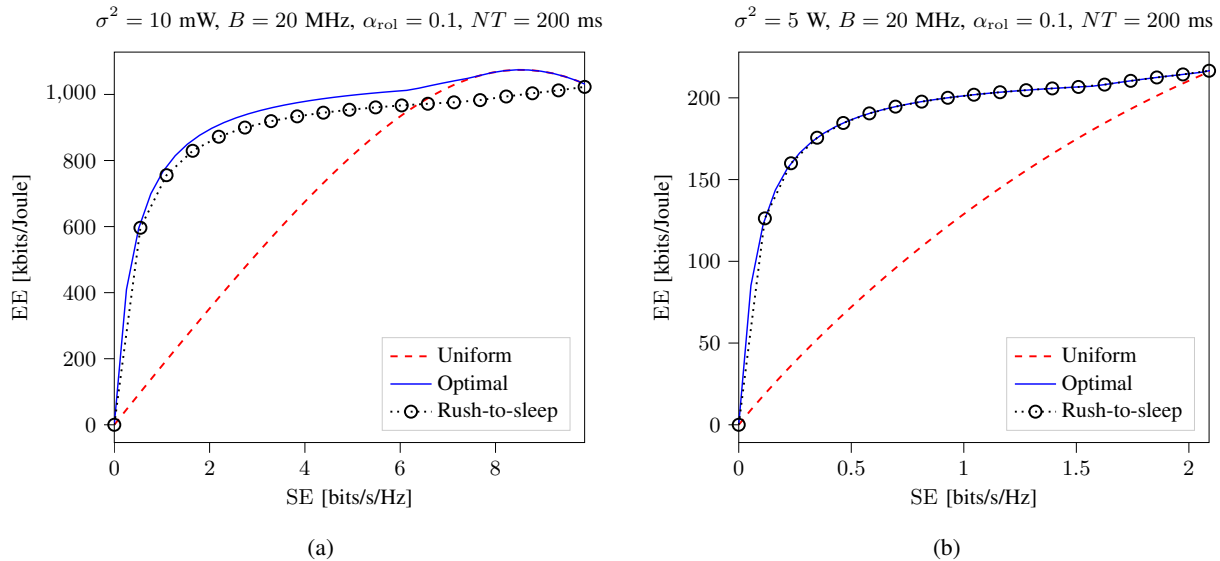
\begin{figure*} 
	\centering
	\subfloat[\label{fig:Fig_7a_SE_EE_P_sleep_decreasing}]{%
		\resizebox{!}{0.415\linewidth}{\footnotesize
			{
\begin{tikzpicture}

\definecolor{darkgray176}{RGB}{176,176,176}
\definecolor{lightgray204}{RGB}{204,204,204}

\begin{axis}[
legend cell align={left},
legend style={
  fill opacity=0.8,
  draw opacity=1,
  text opacity=1,
  at={(0.97,0.03)},
  anchor=south east,
  draw=lightgray204
},
tick align=outside,
tick pos=left,
title={\(\displaystyle \sigma^2=10\) mW, \(\displaystyle B=20\) MHz, \(\displaystyle \alpha_{\mathrm{rol}}=0.1\), \(\displaystyle NT=200\) ms},
x grid style={darkgray176},
xlabel={\(\displaystyle  \mathrm{SE} \) [bits/s/Hz]},
xmin=0, xmax=9.86985391972977,
xtick style={color=black},
y grid style={darkgray176},
ylabel={\(\displaystyle  \mathrm{EE} \) [kbits/Joule]},
ymin=-53.7224005597847, ymax=1128.17041175548,
ytick style={color=black}
]
\addplot [thick, red, dashed]
table {%
9.09090909090909e-05 0.0165265640019834
0.255718075801999 46.126859601197
0.511345242513088 91.8829018482202
0.766972409224178 137.348140436056
1.02259957593527 182.523010522002
1.27822674264636 227.391413060948
1.53385390935745 271.927867525644
1.78948107606854 316.099740571444
2.04510824277963 359.868027902279
2.30073540949071 403.187565954632
2.5563625762018 446.006978758995
2.81198974291289 488.268490986224
3.06761690962398 529.907674127716
3.32324407633507 570.853166625501
3.57887124304616 611.026397799838
3.83449840975725 650.341341261802
4.09012557646834 688.704322676659
4.34575274317943 726.013907640035
4.60137990989052 762.1608971722
4.85700707660161 797.028460364455
5.1126342433127 830.492435579882
5.36826141002379 862.421832900626
5.62388857673488 892.679570794907
5.87951574344597 921.123478794393
6.13514291015706 947.607594855259
6.39077007686815 971.983780562024
6.64639724357924 994.103669011116
6.90202441029033 1013.82094877964
7.15765157700142 1030.99397272291
7.4132787437125 1045.48866259063
7.66890591042359 1057.18166008043
7.92453307713468 1065.96365283217
8.18016024384577 1071.74278130625
8.43578741055686 1074.44801119569
8.69141457726795 1074.03233802962
8.94704174397904 1070.47567814192
9.20266891069013 1063.78729535165
9.45829607740122 1054.00761734188
9.71392324411231 1041.20931100022
9.9695504108234 1025.49751215315
};
\addlegendentry{Uniform}
\addplot [semithick, blue]
table {%
9.09090909090909e-05 0.234553243072798
0.255718075801999 409.836034294481
0.511345242513088 594.366139020389
0.766972409224178 699.353075603991
1.02259957593527 767.108934730254
1.27822674264636 814.455626216912
1.53385390935745 849.407486079236
1.78948107606854 876.268401114231
2.04510824277963 897.556366183816
2.30073540949071 914.842739068478
2.5563625762018 929.158880339719
2.81198974291289 941.20976411939
3.06761690962398 951.493623421911
3.32324407633507 960.372559674311
3.57887124304616 968.116063747182
3.83449840975725 974.928844068548
4.09012557646834 980.969189257148
4.34575274317943 986.361419956696
4.60137990989052 991.204536598576
4.85700707660161 995.578353249269
5.1126342433127 999.547930740865
5.36826141002379 1003.16683493066
5.62388857673488 1006.47956798204
5.87951574344597 1009.52340760732
6.13514291015706 1012.32981590575
6.39077007686815 1018.87563824954
6.64639724357924 1026.9833694934
6.90202441029033 1034.60644192575
7.15765157700142 1041.78705179641
7.4132787437125 1048.56263514675
7.66890591042359 1057.18166008043
7.92453307713468 1065.96365283217
8.18016024384577 1071.74278130625
8.43578741055686 1074.44801119569
8.69141457726795 1074.03233802962
8.94704174397904 1070.47567814192
9.20266891069013 1063.78729535165
9.45829607740122 1054.00761734188
9.71392324411231 1041.20931100022
9.9695504108234 1025.49751215315
};
\addlegendentry{Optimal}
\addplot [thick, black, dotted, mark=o, mark size=2.5, mark options={solid,fill opacity=0}]
table {%
0 0
0.548325217762765 596.367464962667
1.09665043552553 755.589564879087
1.6449756532883 829.402689417218
2.19330087105106 871.995045006454
2.74162608881383 899.716969905551
3.28995130657659 919.198699346215
3.83827652433936 933.638875355012
4.38660174210212 944.770295606815
4.93492695986489 953.61329263797
5.48325217762765 960.807778213787
6.03157739539042 966.775421771207
6.57990261315318 971.80538045915
7.12822783091595 976.102561019162
7.67655304867871 982.826899852043
8.22487826644148 993.747405919376
8.77320348420424 1003.50387872566
9.32152870196701 1012.27301517492
9.86985391972977 1022.84060636205
};
\addlegendentry{Rush-to-sleep}
\end{axis}

\end{tikzpicture}}
	} }
	\hfill
	\subfloat[\label{fig:Fig_7b_SE_EE_P_sleep_decreasing_single_regime}]{%
		\resizebox{!}{0.415\linewidth}{\footnotesize
			{
\begin{tikzpicture}

\definecolor{darkgray176}{RGB}{176,176,176}
\definecolor{lightgray204}{RGB}{204,204,204}

\begin{axis}[
legend cell align={left},
legend style={
  fill opacity=0.8,
  draw opacity=1,
  text opacity=1,
  at={(0.97,0.03)},
  anchor=south east,
  draw=lightgray204
},
tick align=outside,
tick pos=left,
title={\(\displaystyle \sigma^2=5\) W, \(\displaystyle B=20\) MHz, \(\displaystyle \alpha_{\mathrm{rol}}=0.1\), \(\displaystyle NT=200\) ms},
x grid style={darkgray176},
xlabel={\(\displaystyle  \mathrm{SE} \) [bits/s/Hz]},
xmin=0, xmax=2.08973507642512,
xtick style={color=black},
y grid style={darkgray176},
ylabel={\(\displaystyle  \mathrm{EE} \) [kbits/Joule]},
ymin=-10.856382168176, ymax=227.984025531696,
ytick style={color=black}
]
\addplot [thick, red, dashed]
table {%
9.09090909090909e-05 0.016476278915401
0.0542127761046005 9.13653031933736
0.108334643118292 17.7034496949171
0.162456510131983 25.9245463324292
0.216578377145675 33.8698435420598
0.270700244159366 41.5780397950765
0.324822111173058 49.0741821904622
0.378943978186749 56.3759398027155
0.43306584520044 63.4964820545422
0.487187712214132 70.4460091036543
0.541309579227823 77.2326493222015
0.595431446241515 83.8630241320881
0.649553313255206 90.3426230692626
0.703675180268898 96.6760635838633
0.757797047282589 102.867277269151
0.81191891429628 108.919647197223
0.866040781309972 114.836111653944
0.920162648323663 120.619244118897
0.974284515337354 126.271316038096
1.02840638235105 131.7943468668
1.08252824936474 137.190144518827
1.13665011637843 142.460338466497
1.19077198339212 147.606407127227
1.24489385040581 152.629700749507
1.2990157174195 157.531460710776
1.35313758443319 162.312835923049
1.40725945144689 166.974896883504
1.46138131846058 171.518647789308
1.51550318547427 175.945037047286
1.56962505248796 180.254966441533
1.62374691950165 184.44929917012
1.67786878651534 188.528866921748
1.73199065352903 192.494476131521
1.78611252054273 196.346913530083
1.84023438755642 200.086951080393
1.89435625457011 203.715350380489
1.9484781215838 207.232866597665
2.00259998859749 210.640251989014
2.05672185561118 213.938259054729
2.11084372262487 217.12764336352
};
\addlegendentry{Uniform}
\addplot [semithick, blue]
table {%
9.09090909090909e-05 0.2343518612272
0.0542127761046005 85.253905097945
0.108334643118292 122.567327633741
0.162456510131983 143.53064271836
0.216578377145675 156.959028468418
0.270700244159366 166.29607788703
0.324822111173058 173.164413565676
0.378943978186749 178.428804206389
0.43306584520044 182.592353390881
0.487187712214132 185.967672784829
0.541309579227823 188.759244177093
0.595431446241515 191.106358562665
0.649553313255206 193.107484579788
0.703675180268898 194.833811976528
0.757797047282589 196.338307823793
0.81191891429628 197.661148542242
0.866040781309972 198.833360993499
0.920162648323663 199.879286767671
0.974284515337354 200.81829072095
1.02840638235105 201.665971017216
1.08252824936474 202.435032490663
1.13665011637843 203.135927775491
1.19077198339212 203.777283163298
1.24489385040581 204.366469839463
1.2990157174195 204.909562053516
1.35313758443319 205.411765412527
1.40725945144689 205.877530111061
1.46138131846058 206.310683161291
1.51550318547427 206.714533787348
1.56962505248796 207.09195813
1.62374691950165 208.059191551056
1.67786878651534 209.172295380653
1.73199065352903 210.226705220788
1.78611252054273 211.226944221594
1.84023438755642 212.177053990029
1.89435625457011 213.080764218483
1.9484781215838 213.941367819577
2.00259998859749 214.761877202055
2.05672185561118 215.668214678996
2.11084372262487 217.12764336352
};
\addlegendentry{Optimal}
\addplot [thick, black, dotted, mark=o, mark size=2.5, mark options={solid,fill opacity=0}]
table {%
0 0
0.116096393134729 126.268333868647
0.232192786269458 159.980282378065
0.348289179404186 175.60866722574
0.464385572538915 184.626707430415
0.580481965673644 190.496245046556
0.696578358808373 194.621093670728
0.812674751943102 197.678498832016
0.92877114507783 200.035344185534
1.04486753821256 201.907663799079
1.16096393134729 203.430945601109
1.27706032448202 204.694469272985
1.39315671761675 205.759457791422
1.50925311075147 206.669295872008
1.6253495038862 208.093033937346
1.74144589702093 210.405222624921
1.85754229015566 212.470951622664
1.97363868329039 214.327632803258
2.08973507642512 216.565099148421
};
\addlegendentry{Rush-to-sleep}
\end{axis}

\end{tikzpicture}}
	} }
	\caption{Energy efficiency versus spectral efficiency with successive sleep power model~$\mathbf{(As3)}$.}
	\label{fig7} 
    \vspace{-2em}
\end{figure*}

\begin{remark}
	Similarly, Fig.~\ref{fig:Fig_7a_SE_EE_P_sleep_decreasing} and \ref{fig:Fig_7b_SE_EE_P_sleep_decreasing_single_regime} can be compared to Fig.~\ref{fig:Fig_5a_SE_EE} and \ref{fig:Fig_5b_SE_EE_single_regime}, where only the sleep model differs. As ideally expected, the EE quickly reaches a plateau. To still improve this behaviour, a longer frame duration can be used.
\end{remark}


\section{Optimal Allocation for TDMA System}
\label{section:opt_allocation_MU}

%
We now extend previous results by considering a downlink transmission from the \gls{bs} to $K$ users. 
We consider the constant power sleep model so that $(\textbf{As1})-(\textbf{As4})$ hold. The users are multiplexed using \gls{tdma}. In the frame of $N$ symbols, each symbol is allocated to the transmission towards at most one user so that no inter-user interference is present. Out of the $N$ symbols, $N_k$ symbols are allocated to user $k$, for $k=0,...,K-1$. 
The power associated to the $n$-th symbol transmitted to user $k$ is denoted by $p_{k,n}$, where $n=0,...,N_k-1$ and $k=0,...,K-1$. The normalized noise variance $\sigma_k^2=\sigma_n^2L_k$ is considered different at each user, as each can have a specific path loss $L_k$. 
The consumed power is then
\begin{align*}
			P_{\mathrm{cons}}^{\mathrm{TDMA}}&=\frac{N_{\mathrm{a}}}{N}P_0 + \frac{\gamma}{N}\sum_{k=0}^{K-1} \sum_{n=0}^{N_{k-1}}  p_{k,n}^{\alpha}+\frac{N-N_{\mathrm{a}}}{N}P_{\mathrm{sleep}}
\end{align*}
where $N_{\mathrm{a}}=\sum_{k=0}^{K-1} N_k$ so that $0\leq N_{\mathrm{a}}\leq N$. We consider the generalization of problem~(\ref{eq:basic_problem}) of minimizing the power consumption, under $(\textbf{As1})-(\textbf{As4})$ and per-user rate constraints $R_k$. The problem can be formulated as
\begin{align}
	\min_{\substack{N_{k},p_{k,n}\\ n=0,...,N_k-1\\ k=0,...,K-1}} P_{\mathrm{cons}}^{\mathrm{TDMA}} \ \text{s.t. } &\frac{1}{N}\sum_{n=0}^{N_{k}-1} \log_2\left(1+\frac{p_{k,n}}{\sigma_k^2}\right)=R_k\ \forall k,\ \sum_{k=0}^{K-1} N_k\leq N
	.\label{eq:tdma_problem}
\end{align}

We assume in the following that the problem has a feasible solution, which generalizes $\textbf{(As5)}$.

$\textbf{(As6)}$: Problem~(\ref{eq:tdma_problem}) is feasible. Defining the maximal per-user rate $R_{k,\mathrm{max}}=\log_2\left(1+\frac{P_{\max}}{\sigma_k^2}\right)$, it implies that
\begin{align*}
	\sum_{k=0}^{K-1}\ceil{\frac{N R_k}{R_{k,\mathrm{max}}}}\leq  N.
\end{align*}
Moreover, we define 
the constant $R_{k,\mathrm{a}}$ that minimizes the convex problem
\begin{align*}
	R_{k,\mathrm{a}} &= \arg \min_{x \geq 0} \frac{P_0- P_{\mathrm{sleep}} +\gamma \sigma^{2\alpha}_k \left(2^{x}-1\right)^{\alpha}}{x}.
\end{align*}
We also define $\hat{R}_k=\min(R_{k,\mathrm{a}},R_{k,\mathrm{max}})$ and $\hat{P}_k=\min(P_{k,\mathrm{a}},P_{\mathrm{max}})$. Two regimes can be considered for Problem~(\ref{eq:tdma_problem}), depending if the constraint $\sum_{k=0}^{K-1} N_k\leq N$ is binding or not.


In terms of power savings, the most promising case is the low-to-medium load regime where the constraint is not binding. In that case, the problem fully decouples per-user and the linear regime solution of Theor.~\ref{theor:main_SU_result} can directly be used on a per-user basis.

\begin{theorem}[TDMA - linear regime]\label{theor:TDMA_not_binding}
	Under $\mathbf{(As1)}-\mathbf{(As4)},\mathbf{(As6)}$, as $N\rightarrow+\infty$, if $
		\sum_{k=0}^{K-1} \frac{R_k}{\hat{R}_k} \leq 1 +O(1/N)
	$, the allocation
	\begin{align*}
		N_{k}&=\round{NR_k/\hat{R}_k},\ p_{k,n}= \hat{P}_k +O(1/N) \quad \text{for }n=0,...,N_{k}-1
	\end{align*}
	for user $k=0,...,K-1$ is an asymptotic solution of Problem~(\ref{eq:tdma_problem}) and 
	\begin{align*}
		P_{\mathrm{cons}}&=P_{\mathrm{sleep}}+\sum_{k=0}^{K-1}R_k\frac{P_0-P_{\mathrm{sleep}} + \gamma \hat{P}_k^{\alpha}}{\hat{R}_k}+O(1/N).
	\end{align*}
\end{theorem}
\begin{proof}
	See Appendix~\ref{Appendix:proof:theor_TDMA_not_binding}.
\end{proof}

\begin{figure*}
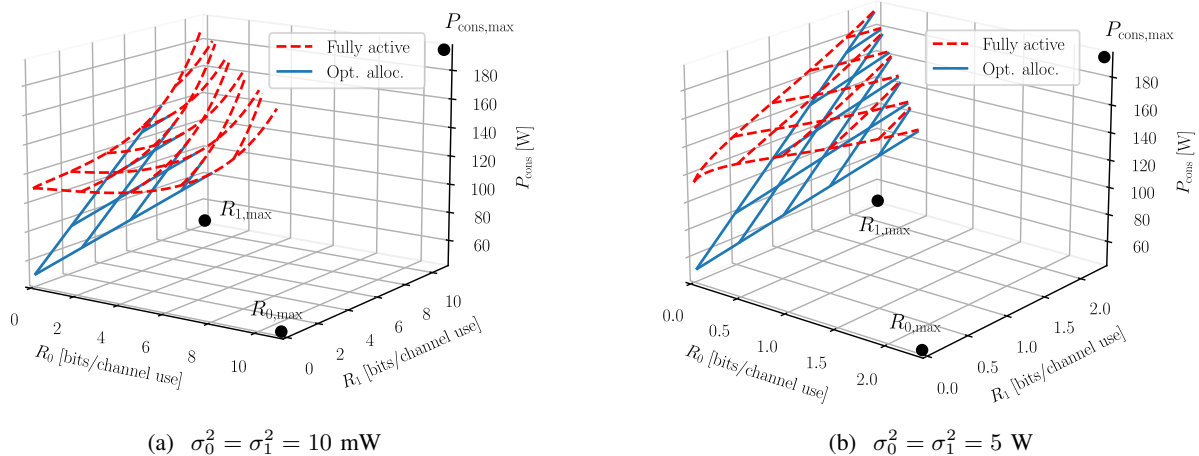
 
	\centering
	\subfloat[\label{fig:Fig_8a_TDMA_linear} $\sigma_0^2=\sigma_1^2=10$ mW]{
		\resizebox{0.45\linewidth}{!}{
			\includesvg[width=0.48\textwidth]{Fig/Python/Fig_8a_TDMA_linear}	
		} 
	}
	\hfill
	\subfloat[\label{fig:Fig_8b_TDMA_linear} $\sigma_0^2=\sigma_1^2=5$ W]{%
		\resizebox{0.45\linewidth}{!}{
			\includesvg[width=0.48\textwidth]{Fig/Python/Fig_8b_TDMA_linear}	
	} }
	\caption{Power consumption versus load in a $K=2$ user TDMA system using optimal/uniform power allocation with constant sleep mode~$\mathbf{(As4)}$. The optimal allocation is valid in the asymptotic linear regime, \textit{i.e.}, when $\frac{R_0}{\hat{R}_0}+\frac{R_1}{\hat{R}_1} \leq 1$.}
	\label{Fig_8} 
    \vspace{-2em}
\end{figure*}

\begin{remark}[$K=2$ - TDMA]
	For $K=2$, Fig.~\ref{fig:Fig_8a_TDMA_linear} and \ref{fig:Fig_8b_TDMA_linear} plot $P_{\mathrm{cons}}$ as a function of $R_1$ and $R_2$. The optimal solution of Theor.~\ref{theor:TDMA_not_binding} is plotted where it is valid, \textit{i.e.}, the asymptotic linear regime $\frac{R_0}{\hat{R}_0}+\frac{R_1}{\hat{R}_1}\leq 1$. As already observed in the single-user case (Fig.~\ref{fig:Fig_4b_scaling_regimes_only_linear}), this regime can cover the whole feasible rate region, as shown in Fig.~\ref{fig:Fig_8b_TDMA_linear}, characterized by a relatively higher noise power implying that $\hat{R}_k=R_{k,\mathrm{max}}$, $\forall k$. As a benchmark, a fully active uniform allocation was plotted where a share $N_k/N=R_k/(R_0+R_1)$ of time slots was allocated to each user.
 \end{remark}
 \begin{remark}[rush-to-sleep - TDMA]
     If $\sum_{k=0}^{K-1} \frac{R_k}{\hat{R}_k}\leq 1$, a rush-to-sleep allocation is asymptotically optimal for each user such that $R_{k,\mathrm{max}} \leq R_{k,\mathrm{a}}$. On the other hand, if $\sum_{k=0}^{K-1} \frac{R_k}{\hat{R}_k}>1$, it is optimal to use a fully active system and no sleep. 
 \end{remark}


If the constraint $\sum_{k=0}^{K-1} N_k = N$ is binding, Problem~(\ref{eq:tdma_problem}) is coupled between users and challenging to solve. As the sleep duration is zero, the power consumption becomes
\begin{align*}
	P_{\mathrm{cons}}^{\mathrm{TDMA}}&=P_0 + \frac{\gamma}{N}\sum_{k=0}^{K-1} \sum_{n=0}^{N_{k-1}}  p_{k,n}^{\alpha}.
\end{align*}
As this regime does not allow switching-off components, relatively small energy reduction potentials are expected. Given space constraints, we do not provide a more detailed solution. One possibility is to reduce the target rates of the users to make the constraint non-binding and then use Theor.~\ref{theor:TDMA_not_binding}. Another possibility is to use conventional scheduling policies, that are not aware of sleep capabilities, which makes sense as the system is fully active. 

\section{Conclusion}
\label{section:conclusion}

In this work, we have proposed a fundamental study of time-domain energy-saving techniques in radio access. The results provide key novel insights from an information-theoretic perspective. Considering equal gain parallel communication channels, conventional information-theoretic results state that all channels (time slots in this study) should be equally used to minimize transmit (not consumed) power under a rate constraint. On the contrary, popular energy-saving techniques steer towards an extreme opposite ``rush-to-sleep" approach: compact transmission in as few time slots as possible, at maximal transmit power, to maximize sleep duration. 

Using a realistic power consumption model, our information-theoretic study bridges the gap between these two extremes. Simple allocations are provided that allow drastic energy savings reaching factors of 10 at low load. At low-to-medium load, the optimal number of active time slots is linearly proportional to the rate, resulting in a power consumption which linearly scales with the rate. At a higher load, all time slots become allocated. The rush-to-sleep approach is shown to be optimal in a high-noise regime but not otherwise. In a low-noise regime, it might be better to use a fully active system. Moreover, the fundamental trade-off between \gls{ee} and \gls{se} is revisited leveraging the time-domain hardware sleep capabilities. Transmitting at maximal \gls{se} maximizes the \gls{se} in a high-noise regime while, in a low-noise regime, a reduced \gls{se} maximizes the \gls{ee}. Considering a sleep model with increasing depth complicates the study but also greatly increases the energy-saving gains. For a piecewise constant model, simple allocations can still be found. Finally, for a multi-user \gls{tdma} system, single-user results are applicable on a per-user basis in the low-to-medium load regime, where the system should not be fully active and where sleep-aware energy-saving gains can be achieved.

\section{Appendix}
\label{section:appendix}

We start by introducing two lemmas that will be useful in the following.

\begin{lemma}\label{lemma:uniform_should_be}
	An optimal solution of the following problem
	\begin{align}
		\min_{p_0,...,p_{N-1}} \frac{\gamma}{N}\sum_{n=0}^{N-1} p_n^\alpha \quad \text{s.t. } \frac{1}{N}\sum_{n=0}^{N-1} \log_2\left(1+\frac{p_n}{\sigma^2}\right) = R \label{eq:prob_lemma_PA_var}
	\end{align}
	must have a uniform allocation among active time slots: $\forall n,n'$, if $p_n>0$, $p_{n'}>0$ then $p_n=p_{n'}$.
\end{lemma}
\begin{proof} The case $\alpha=1$ was already treated in the introduction.
For $\alpha=0$, the cost function only depends on the number of active time slots so that a single time slot must be allocated power. In the following, we will use a proof by contradiction for the case $0<\alpha<1$. 
Any non-uniform power allocation has at least two time slots, say $n=0$ and $n=1$ (potentially using a re-indexing), that are such that 
$p_0>0$, $p_1>0$ and $p_0\neq p_1$. We consider the power allocated to the other time slots as fixed and optimize the cost function with respect to $p_0$ and $p_1$ only. 
For the sake of clarity, we define $\rho_n=p_n/\sigma^2$. The reduced problem is
\begin{align*}
	\min_{\rho_0,\rho_{1}} \tilde{f}(\rho_0,\rho_1)=\sum_{n=0}^{1} \rho_n^{\alpha} \ \text{s.t. }  \sum_{n=0}^{1} \log\left(1+\rho_n\right) = Z
\end{align*}
where $Z=NR \log 2 - \sum_{n=2}^{N-1} \log\left(1+\rho_n\right)$. The constraint implies that
\begin{align}
	\rho_1&=\frac{e^Z}{1+{\rho_0}}-1 \label{eq:relationship_rho_1_rho_0}
\end{align}
so that $\rho_0$ and $\rho_1$ take values only in the domain $[0,e^Z-1]$ and have a one-to-one relationship. As $\rho_0 \rightarrow 0$, $\rho_1 \rightarrow e^Z-1$ and vice versa. Moreover, the problem is symmetrical so that $\tilde{f}(\rho_0,\rho_1)=\tilde{f}(\rho_1,\rho_0)$. Using~(\ref{eq:relationship_rho_1_rho_0}), the problem can be rewritten as a monovariable unconstrained problem
\begin{align*}
	\min_{\rho_{0}} \tilde{f}(\rho_0)= \rho_0^\alpha+\rho_1^\alpha= \rho_0^\alpha+\left(\frac{e^Z}{1+{\rho_0}}-1\right)^\alpha.
\end{align*}
The derivative of $\tilde{f}(\rho_0)$ and its limit at the bounds of its domain are given by
\begin{align}
	\tilde{f}'(\rho_0)
	&=\frac{\alpha}{\rho_0^{1-\alpha}}-\frac{e^Z}{(1+\rho_0)^2}\frac{\alpha}{\rho_1^{1-\alpha}},\
	\lim_{\rho_0\rightarrow 0}\tilde{f}'(\rho_0)=+\infty,\ 	\lim_{\rho_0\rightarrow e^{Z}-1} \tilde{f}'(\rho_0)=-\infty. \label{eq:FOC_bounds_dom_definition}
\end{align}
To find the critical points, we set $f'(\rho_0)=0$ and combining with~(\ref{eq:relationship_rho_1_rho_0}), we find the condition
\begin{align}\label{eq:FOC}
	0&=\rho_1^{1-\alpha}(1+\rho_0)-\rho_0^{1-\alpha}(1+\rho_1)
\end{align}
which shows that there is always one critical point in $\rho_0=\rho_1=e^{Z/2}-1$. Moreover, the fact that the problem is symmetric implies an odd number of critical points: if there is a critical point in $\tilde{\rho}_0$, there is one in $\frac{e^Z}{1+{\tilde{\rho}_0}}-1$. Using again~(\ref{eq:relationship_rho_1_rho_0}), we can rewrite the condition~(\ref{eq:FOC}) as
\begin{align*}
	0&=(e^{Z}-1-\rho_0)^{1-\alpha}(1+\rho_0)^{1+\alpha}-\rho_0^{1-\alpha}e^Z.
\end{align*}
The point $\rho_0=0$ is not a critical point. Hence, we can restrict to $\rho_0>0$ and divide by $\rho_0^{1-\alpha}$
\begin{align*}
	0&=\underbrace{(e^{Z}-1-\rho_0)^{1-\alpha}(1+\rho_0)^{1+\alpha}\rho_0^{\alpha-1}-e^Z}_{\tilde{g}(\rho_0)}.
\end{align*}
The function $\tilde{g}(\rho_0)$ is infinite in $\rho_0=0$ and $-e^Z$ for $\rho_0=e^Z-1$. The number of roots of $\tilde{g}(\rho_0)$ and thus critical points of $\tilde{f}(\rho_0)$ is at most equal to one plus the number of critical points/alternations of $\tilde{g}(\rho_0)$. Setting $\tilde{g}'(\rho_0)=0$ gives the condition
\begin{align*}
	0=-(1-\alpha)(1+\rho_0)\rho_0+(e^{Z}-1-\rho_0)(1+\alpha)\rho_0+(e^{Z}-1-\rho_0)(1+\rho_0)(\alpha-1)
\end{align*}
which is a quadratic equation in $\rho_0$. It has thus at most two solutions. As a result, $\tilde{g}(\rho_0)$ has max two alternations and $\tilde{f}(\rho_0)$ has at most 3 critical points. The second order derivative of $\tilde{f}(\rho_0)$ and its limit at the bounds of its domain are given by
\begin{align}
	\tilde{f}''(\rho_0)
	&=-\alpha(1-\alpha) \frac{1}{\rho_0^{2-\alpha}}+\frac{1}{\rho_1^{2-\alpha}} \frac{\alpha e^Z}{(1+\rho_0)^3} \left(-(1-\alpha)  e^{Z}\frac{1}{1+\rho_0} + 2\rho_1 \right) \nonumber \\
	\lim_{\rho_0\rightarrow 0}\tilde{f}''(\rho_0)&=-\infty,\ 
	\lim_{\rho_0\rightarrow e^{Z}-1}\tilde{f}''(\rho_0)=-\infty. \label{eq:SOC_bounds_dom_definition}
\end{align}

\begin{figure*}[t!]	
	\begin{minipage}{.6\textwidth}
		\resizebox{1\linewidth}{!}{%
			{\includesvg{Fig//Inkscape/f_tilde}} 
		}
	\end{minipage}
	\hfill
	\begin{minipage}{.3\textwidth}
		\resizebox{1\linewidth}{!}{%
			\begin{tabular}{ |c|c|  }
				\hline
				Domain & $[0,e^{Z}-1]$  \\
				\hline
				Intercept & $\tilde{f}(0)=(e^Z-1)^\alpha$ \\
				\hline
				Symmetry & $\tilde{f}(\rho_0)=\tilde{f}\left(\frac{e^Z}{1+{\rho_0}}-1\right)$  \\
				\hline
				Derivatives & $\tilde{f}'(0)=+\infty$, $\tilde{f}'(e^{Z}-1)=-\infty$ \\
				\hline
				Concavity &  $\tilde{f}''(0)=\tilde{f}''(e^{Z}-1)=-\infty$\\
				\hline
				Critical points & Always 1 in $\rho_0=\rho_1=e^{Z/2}-1$ \\
				& Potentially 2 others symmetrical\\
				\hline
			\end{tabular}
		}
	\end{minipage}
	\caption{Sketch of function $\tilde{f}(\rho_0)$.}
	\label{fig:f_tilde} 
	\vspace{-2em}
\end{figure*}

%


As shown in Fig.~\ref{fig:f_tilde}, three cases can be distinguished. In case (a), there is a single critical point in $\rho_0=\rho_1=e^{Z/2}-1$ which is a maximum. In cases (b) and (c), there are three critical points: the middle one in $\rho_0=\rho_1=e^{Z/2}-1$ will now be a minimum (local in (b), global in (c)) while the two on its sides are maxima. Hence, global minima can only be obtained for either $\rho_{0}=0,\rho_{1}=e^{Z}-1$ or $\rho_{1}=0,\rho_{0}=e^{Z}-1$ or $\rho_0=\rho_1=e^{Z/2}-1$. Hence, it is impossible to find an optimal allocation such that $p_0>0$, $p_1>0$ and $p_0\neq p_1$. \end{proof}



\begin{lemma}\label{lemma:convexity}
	Under $\mathbf{(As1)}-\mathbf{(As4)}$, the following function is convex for $x>0$
	\begin{align*}
		f(x)&=\frac{P_0-P_{\mathrm{sleep}} + \gamma \sigma^{2\alpha} \left(2^{x}-1\right)^{\alpha}}{x}.
	\end{align*}
\end{lemma}
\begin{proof}
	Under $\mathbf{\mathbf{(As4)}}$, $P_0-P_{\mathrm{sleep}}\geq 0$ and thus $(P_0-P_{\mathrm{sleep}})/x$ is convex. Given that the sum of two convex functions is convex, it is sufficient to show that
	\begin{align*}
		\frac{\gamma \sigma^{2\alpha} \left(2^{x}-1\right)^{\alpha}}{x} \text{ or equivalently } g(y)=\frac{\left(e^{y}-1\right)^{\alpha}}{y}
	\end{align*}
	is convex for $y>0$ (using $y=x \log 2$). Its second derivative is
	$$g''(y) = \frac{(e^y - 1)^{\alpha - 2}(e^{2y} (\alpha^2y^2 - 2 \alpha y + 2) - e^y (\alpha y^2 - 2\alpha y + 4) + 2)}{y^3}.$$
	Given that $y>0$, we have directly that $y^3>0$, $(e^y - 1)^{\alpha - 2}\geq 0$ and it is sufficient to show that
	$$h(y)=e^{2y} (\alpha^2y^2 - 2 \alpha y + 2) - e^y (\alpha y^2 - 2\alpha y + 4) + 2 \geq 0.$$

	Using the Taylor series expansion $e^y=\sum_{r=0}^{+\infty}\frac{y^r}{r!}$, which converges for all $y$, we find
	\begin{align*}
		h(y)&=\sum_{r=0}^{+\infty}\frac{(2x)^r}{r!} (\alpha^2y^2 - 2 \alpha y + 2) - \sum_{r=0}^{+\infty}\frac{y^r}{r!} (\alpha y^2 - 2\alpha y + 4) + 2\\
		&= \sum_{r=2}^{+\infty}y^{r}\left(\alpha\frac{\alpha 2^{r-2} -1}{(r-2)!} + \alpha \frac{2-2^{r}}{(r-1)!} + \frac{2^{r+1}-4}{r!} \right).
	\end{align*}
	To show that $h(y)\geq 0$ for $y>0$ and $\alpha \in [0,1]$, it is sufficient to show that for all $r\geq 2$
	\begin{align*}
		\alpha\frac{2^{r-2}\alpha  - 1}{(r-2)!}+\alpha\frac{2-2^{r}}{(r-1)!}+\frac{2^{r+1}-4}{r!} \geq 0\\
		\leftrightarrow \alpha^2 2^{r-2}r(r-1)+\alpha r(3-2^{r}-r)+2^{r+1}-4 &\geq 0.
	\end{align*}
	For $r=2$, this is verified as $2\alpha^2 -6\alpha +4 = 2(\alpha-1)(\alpha-2)$
	is always positive for $\alpha \in [0,1]$. For $r=3$, this is also verified as $\alpha^2 12 -\alpha 24 + 12  = 12 (\alpha-1)^2$
	is again positive for $\alpha \in [0,1]$. For $r\geq 4$, we have $r(r-1)\geq r^2/2$, $3-2^{r}-r\geq -2^{r+1}$ and $-2^{-r+4}\geq -1$ so that
	\begin{align*}
		\alpha^2 2^{r-2}r(r-1)+\alpha r(3-2^{r}-r)+2^{r+1}-4 
		&\geq 2^{r-2}( \alpha^2  r^2/2 -\alpha r 8+8-2^{-r+4})\\
		&\geq 2^{r-2}( \alpha^2  r^2/2 -\alpha r 8+7)
	\end{align*}
	which roots are in $8r\pm \sqrt{50} r$. Given that $8r - \sqrt{50} r \geq 0.92 r > 1 $ for $r\geq 2$, both roots are strictly larger than 1 and the term is positive for $\alpha \in [0,1]$, which concludes the proof.
\end{proof}

\vspace{-1em}\subsection{Proof of Lemmas \ref{lemma:PA_var_problem} and \ref{lemma:no_max_problem}}\label{Appendix:proof_lemma_no_max_problem}

One can first note that Lemma~\ref{lemma:PA_var_problem} is a particularization of Lemma~\ref{lemma:no_max_problem} when $P_0=E_{\mathrm{sleep}}(t)=0$. Hence, the result will be found as a specific case in the following. Under $\mathbf{(As1)}$-$\mathbf{(As2)}$, Problem~(\ref{eq:basic_problem}) can be rewritten as
\begin{align}
	\min_{N_{\mathrm{a}}} \frac{N_{\mathrm{a}}}{N}P_0  + \left[\min_{p_0,...,p_{N_{\mathrm{a}}-1}} \frac{\gamma}{N} \sum_{n=0}^{N_{\mathrm{a}-1}}  p_n^{\alpha}\right] +\frac{E_{\mathrm{sleep}}((N-N_\mathrm{a})T)}{NT}
\end{align}
which shows that only the second term depends on the power allocation, \textit{i.e.}, the term defined as $P_{\mathrm{ld}}$. From Lemma~\ref{lemma:uniform_should_be},  an optimal allocation needs to be uniform in the number of activated time slots. For a given $N_{\mathrm{a}}$, the rate constraint fixes the transmit power per active time slot
\begin{align*}
	p_n&=\left(2^{R\frac{N}{N_{\mathrm{a}}}}-1\right)\sigma^2 \text{ if } n=0,...,N_{\mathrm{a}}-1.
\end{align*}
Hence, the problem can be reformulated as finding the optimal number of active slots $N_{\mathrm{a}}$ that minimizes the consumed power, \textit{i.e.}, Problem~(\ref{eq:theor_no_max_problem}) in Lemma~\ref{lemma:no_max_problem}. Moreover, under $\mathbf{(As3)}-\mathbf{(As4)}$, the problem becomes
\begin{align*}
	\min_{N_{\mathrm{a}}} P_{\mathrm{sleep}}+ \frac{N_{\mathrm{a}}}{N}\left(P_0-P_{\mathrm{sleep}} + \gamma \sigma^{2\alpha} \left(2^{\frac{NR}{N_{\mathrm{a}}}}-1\right)^{\alpha}\right).
\end{align*}	
We define $x=\frac{RN}{N_{\mathrm{a}}}$ and relax the problem by considering $x$ as continuous
\begin{align*}
	\min_{x} f(x)&=\frac{P_0-P_{\mathrm{sleep}} + \gamma \sigma^{2\alpha} \left(2^{x}-1\right)^{\alpha}}{x}.
\end{align*}
From Lemma~\ref{lemma:convexity}, we know that $f(x)$ is convex. 
Given the definition of $x$ and the integer nature of $N_{\mathrm{a}}$, $x$ can only take discrete values in practice. Given the fact that $f(x)$ is convex, it is guaranteed that one of the neighboring possible values of $R_{\mathrm{a}} = \arg \min f(x)$ is optimal. As a result, the solution is given by either $N_{\mathrm{a}}=\ceil{\frac{RN}{R_{\mathrm{a}}}}$, $N_{\mathrm{a}}=\floor{\frac{RN}{R_{\mathrm{a}}}}$ or $N$ if $\frac{RN}{R_{\mathrm{a}}}>N$. Using the ceil-floor notation concludes the proof of Lemma~\ref{lemma:no_max_problem}. The above result can also be particularized to the problem of Lemma~\ref{lemma:PA_var_problem} by setting $P_0=E_{\mathrm{sleep}}(t)=0$ and the problem simplifies to
\begin{align*}
	\min_{x} f(x)&=\frac{\left(2^{x}-1\right)^{\alpha}}{x} \leftrightarrow \min_{y} g(y)=\frac{\left(e^{y}-1\right)^{\alpha}}{y}.
\end{align*}
where $y=x \log 2$. Setting its derivative to zero, we find
\begin{align*}
	\alpha e^{y} y&=e^{y}-1 \leftrightarrow y=W\left(-\alpha^{-1}e^{-\alpha^{-1}}\right)+\alpha^{-1} \leftrightarrow R_{\mathrm{a}}=W\left(-\alpha^{-1}e^{-\alpha^{-1}}\right)+\alpha^{-1}/\log 2,
\end{align*}
which concludes the proof of Lemma~\ref{lemma:PA_var_problem}.

%

\vspace{-1em}\subsection{Proof of Proposition~\ref{prop:optimal}}\label{Appendix:proof_prop_optimal}

The algorithm is initialized by the solution of the relaxed problem assuming that no max power constraints are binding. The solution is then given by~Theor.~\ref{lemma:no_max_problem}. If the solution is such that $p_n \leq P_{\mathrm{max}}, \forall n$, the problem is solved. On the other hand, if, for at least one time slot $p_n>P_{\mathrm{max}}$, at least one of the constraints must be binding. Hence, the algorithm allocates the maximal power $P_{\mathrm{max}}$ to one additional time slot. The rate constraint on the remaining time slots is then adapted. The power allocation is re-computed assuming that no max power constraint is binding on the remaining time slots not yet set to $P_{\mathrm{max}}$. Again, the solution is given by~Theor.~\ref{lemma:no_max_problem}. Again, the algorithm checks if the max constraint is verified. If yes, the algorithm has converged. If not, it enters a novel iteration and so on until convergence.

\vspace{-1em}\subsection{Proof of Theorem~\ref{theor:main_SU_result}}\label{Appendix:proof_theorem_main_SU_result}

Let us consider one by one four different cases. On the one hand, the exponential regime mentioned in the theorem where $R>\tilde{R}$ and i) $\tilde{R} = R_{\mathrm{max}}$ or ii) $\tilde{R} = R_{\mathrm{a}}$. On the other hand, the linear regime mentioned in the theorem where $R\leq\tilde{R}$ and iii) $\tilde{R} = R_{\mathrm{a}}$ or iv) $\tilde{R} = R_{\mathrm{max}}$.

\underline{Case i)}: This case is not applicable according to $\mathbf{(As5)}$ as it is unfeasible to have $R > R_{\mathrm{max}}$.

\underline{Case ii)}: The case together with $\mathbf{(As5)}$ implies $R_{\mathrm{max}} \leq  R > R_{\mathrm{a}}$. From Lemma~\ref{lemma:no_max_problem}, we can find that $N_{\mathrm{a}}=N$ and the corresponding power allocation, which is well feasible as it does not violate the $P_{\mathrm{max}}$ constraint. The exponential regime result of Theor.~\ref{theor:main_SU_result} is then found.

\underline{Case iii)}: This case implies $R \leq R_{\mathrm{a}} \leq R_{\mathrm{max}}$. Let us first consider that the maximal power constraint per time slot is not active such that we can use the result of Lemma~\ref{lemma:no_max_problem}. If $R \leq R_{\mathrm{a}}$, the optimal number and ratio of active time slots are
\begin{align*}
	N_{\mathrm{a}}&=\ceilfloor{{RN}/{R_{\mathrm{a}}}}=\round{\frac{RN}{R_{\mathrm{a}}}}+\epsilon_1=\frac{RN}{R_{\mathrm{a}}}+\epsilon_2\\
	\frac{N_{\mathrm{a}}}{N}&=\round{\frac{RN}{R_{\mathrm{a}}}}/N+\frac{\epsilon_1}{N}=\frac{R}{R_{\mathrm{a}}}+\frac{\epsilon_2}{N}.
\end{align*}
where $|\epsilon_1|<1$ and $|\epsilon_2|<1$. The optimal ratio ${N_{\mathrm{a}}}/{N}$ asymptotically converges to $R/R_{\mathrm{a}}$ and the same occurs if the optimal number of time slots $N_{\mathrm{a}}$ is approximated using a rounding operator instead of the ceil-floor operator. Using this result, as $N \rightarrow +\infty$, the optimal power allocation per active time slot of Lemma~\ref{lemma:no_max_problem} can be rewritten as 
\begin{align}
	p_n&=\left(2^{R\frac{N}{N_{\mathrm{a}}}}-1\right)\sigma^2=P_{\mathrm{a}} + O(1/N) \label{eq:asymptotic_opt}
\end{align}
where $P_{\mathrm{a}}=\left(2^{R_{\mathrm{a}}}-1\right)\sigma^2$ and $N_{\mathrm{a}}$ can be the ideal value $\ceilfloor{{RN}/{R_{\mathrm{a}}}}$ or its approximation using the rounding operator. Given that $R_{\mathrm{a}} \leq R_{\mathrm{max}}$, this allocation does not violate the $P_{\mathrm{max}}$ constraint and the result is feasible. The power consumption becomes
\begin{align*}
	P_{\mathrm{cons}}&=P_{\mathrm{sleep}}+ \frac{N_{\mathrm{a}}}{N}\left(P_0-P_{\mathrm{sleep}} + \gamma P_{\mathrm{a}}^{\alpha} + O(1/N)\right)=P_{\mathrm{sleep}}+R\frac{P_0-P_{\mathrm{sleep}} + \gamma P_{\mathrm{a}}^{\alpha}}{R_{\mathrm{a}}}+O(1/N).
\end{align*}

\underline{Case iv)}: This case implies $R \leq R_{\mathrm{max}} \leq R_{\mathrm{a}}$ and thus $P_{\mathrm{a}}>P_{\mathrm{max}}$ such that allocation~(\ref{eq:asymptotic_opt}) is not feasible. As an alternative, the iterative Algorithm~\ref{alg:max_constraint} can be used and simplified in the asymptotic regime. Indeed, as $N \rightarrow +\infty$, at each iteration, the algorithm allocates a constant power $P_{\mathrm{a}}$ (independent of $R$) to active time slots, not yet set to $P_{\mathrm{max}}$. At the convergence of the algorithm, the allocation will have approximately ${RN}/{R_{\mathrm{max}}}$ active time slots with maximal power $P_{\mathrm{max}}$ and rate $R_{\mathrm{max}}$. As a result, as $N \rightarrow +\infty$, at the optimum, $N_{\mathrm{a}}/N=\round{{R}/{R_{\mathrm{max}}}}+O(1/N)={{R}/{R_{\mathrm{max}}}}+O(1/N)$ and the allocation
\begin{align*}
	N_{\mathrm{a}}&=\round{\frac{RN}{R_{\mathrm{max}}}},\ p_n=P_{\mathrm{max}} \text{ for $n=0,...,N_{\mathrm{max}}$ }
\end{align*}
is asymptotically optimal and achieves a consumed power
\begin{align*}
    P_{\mathrm{cons}}&=P_{\mathrm{sleep}}+R\frac{P_0-P_{\mathrm{sleep}} + \gamma P_{\mathrm{a}}^{\alpha}}{R_{\mathrm{max}}}+O(1/N).
\end{align*}
Cases iii) and iv) can be written more compactly using the definitions of $\tilde{R}$ and $\tilde{P}$, giving the linear regime result of Theor.~\ref{theor:main_SU_result}.

\vspace{-1em}\subsection{Proof of Corollary~\ref{corol:max_EE}}\label{proof:corol_max_EE}
The results of Corol.~\ref{corol:max_EE} can be found by minimizing the \gls{ee} expression in the two regimes of Corol.~\ref{corol:EE_SE_trade_off}. In the regime where $\mathrm{SE} \leq \tilde{R}/(1+\alpha_{\mathrm{rol}})$, it is clear that the \gls{ee} is maximized for the largest \gls{se}, \textit{i.e.}, when $\mathrm{SE} = \tilde{R}/(1+\alpha_{\mathrm{rol}})$. Moreover, if $R_{\mathrm{a}} \geq R_{\mathrm{max}}$, we have $\tilde{R}=R_{\mathrm{max}}$ and the second regime of Corol.~\ref{corol:EE_SE_trade_off} is not feasible. The optimal SE corresponds to the maximal SE, $\mathrm{SE}_{\mathrm{max}}=R_{\mathrm{max}}/(1+\alpha_{\mathrm{rol}})$. On the other hand, if $R_{\mathrm{a}} < R_{\mathrm{max}}$, $\tilde{R}=R_{\mathrm{a}}$ and the regime $\mathrm{SE}>\tilde{R}/(1+\alpha_{\mathrm{rol}})$ can be entered. The optimization over $\bar{R}$ then provides the optimum and can only improve the optimum as $\bar{R}$ is allowed to take value $R_{\mathrm{a}}$.

\vspace{-1em}\subsection{Proof of Theorem~\ref{theor:successive_sleep_mode}}\label{Appendix:proof_theorem_successive_sleep_mode}

It is direct to see that one should choose the optimal allocation among feasible sleep modes. Under $\mathbf{(As3)}$, the sleep energy consumption at time $t$ if sleep mode $s$ is used is $E_{\mathrm{sleep},s}(t)=E_{\mathrm{sleep}}(T_s)+(t-T_s)P_s$
where $E_{\mathrm{sleep}}(T_s)=\sum_{s'=0}^{s-1}P_{s'}(T_{s'+1}-T_{s'})$.\footnote{No deeper sleep mode than $s$ is considered even if $T_{s+1}<t$, which could decrease sleep energy consumption. Still, this does not affect the optimization result a deeper sleep mode will perform better and will be chosen instead
	.} The consumed power using sleep mode $s$ can then be written as
\begin{align*}
	P_{\mathrm{cons},s}&=\frac{\tilde{E}_s}{NT}+\frac{N_{\mathrm{a}}}{N}P_0 + \frac{\gamma}{N} \sum_{n=0}^{N_{\mathrm{a}-1}}  p_n^{\alpha} +\frac{N-N_{\mathrm{a}}}{N}P_{s}
\end{align*}
where $\tilde{E}_s=E_{\mathrm{sleep}}(T_s)-T_sP_s$. This form is similar to the one given in (\ref{eq:P_cons_As4}), under $\mathbf{(As4)}$. The sole differences are the presence of $P_{s}$ instead of $P_{\mathrm{sleep}}$ and the constant $\tilde{E}_s/(NT)$, which affects the cost function but does not impact the optimization. The result of Theorem~\ref{theor:main_SU_result} can then be used: uniform allocation among $N_{\mathrm{a},s}$ active mode is optimal. 
The only difference is the fact that the maximal value of $N_{\mathrm{a},s}$ is $N_{\mathrm{a},s}^+$ instead of $N$, so that the sleep duration is sufficient to enter mode $s$. As a result, we find $N_{\mathrm{a},s}=\round{\min\left(NR/\tilde{R}_s,N_{\mathrm{a},s}^+\right)}$.

\vspace{-1em}\subsection{Proof of Theorem~\ref{theor:TDMA_not_binding}}\label{Appendix:proof:theor_TDMA_not_binding}

If the constraint $\sum_{k=0}^{K-1} N_k\leq N$ is not binding, Problem~(\ref{eq:tdma_problem}) is fully decoupled between users and can be solved by solving for $k=0,...,K-1$ an independent per-user problem
\begin{align*}
	\min_{\substack{N_{k},p_{k,n}\\ n=0,...,N_k-1}} &\frac{N_{k}}{N}P_0 + \frac{\gamma}{N} \sum_{n=0}^{N_{k-1}}  p_{k,n}^{\alpha}+\frac{N-N_{k}}{N}P_{\mathrm{sleep}} \text{   s.t.   } \frac{1}{N}\sum_{n=0}^{N_{k}-1} \log_2\left(1+\frac{p_{k,n}}{\sigma_k^2}\right)=R_k
\end{align*}
so that the asymptotic solution of Theorem~\ref{theor:main_SU_result} can be used giving $N_{k}=\round{NR_k/\hat{R}_k}$. If the constraint $\sum_{k=0}^{K-1} N_k\leq N$ is not violated, this is the asymptotic solution of Problem~(\ref{eq:tdma_problem}). From Section~\ref{Appendix:proof_theorem_main_SU_result}, we know that, as $N\rightarrow +\infty$, $N_k/N=R_k/\hat{R}_k+O(1/N)$ and the constraint can thus be equivalently written as $\sum_{k=0}^{K-1} R_k/\hat{R}_k\leq 1 + O(1/N)$.

\vspace{-0.5em}
\section*{Acknowledgment}

The author would like to thank his colleagues from the DRAMCO-KU Leuven lab and Dr.~Pål Frenger for many fruitful discussions and valuable comments.

\ifCLASSOPTIONcaptionsoff
  \newpage
\fi



\bibliographystyle{IEEEtran}

\vspace{-0.5em}
\bibliography{IEEEabrv,refs}

\end{document}